\documentclass[11pt]{article}

\usepackage[
	persons={Nima, Yang, June},
	bibliosources={refs.bib},
]{pomegranate}

\DeclareDelimiter{\Otilde}[\widetilde{O}]{\lparen}{\rparen}
\DeclareDelimiter{\Var}[\operatorname{Var}]{\lbrack}{\rbrack}
\DeclareDocumentMathCommand{\DKL}{s}{\IfBooleanTF{#1}{\D*_{\operatorname{KL}}}{\D_{\operatorname{KL}}}}

\DeclareDocumentMathCommand{\bmu}{}{\bar{\mu}}
\DeclareDocumentMathCommand{\bnu}{}{\bar{\nu}}
\DeclareOperator{\TV}
\DeclareOperator{\mix}

\DeclareDocumentCommand{\funding}{}{
	Nima Anari and Thuy-Duong Vuong are supported by NSF CAREER Award CCF-2045354, a Sloan Research Fellowship, and a Google Faculty Research Award. Yang P. Liu was supported by the Department of Defense (DoD) through the National Defense Science and Engineering Graduate Fellowship, and NSF CAREER Award CCF-1844855 and NSF Grant CCF-1955039.
}
\DeclareDocumentCommand{\acks}{}{
	We thank Michal Derezi\'nski and Elizabeth Yang for useful discussions.
}

\Tag{
	\title{Optimal Sublinear Sampling of\\ Spanning Trees and Determinantal Point Processes via\\ Average-Case Entropic Independence}
	\author{Nima Anari}
	\author{Yang P.\ Liu}
	\author{Thuy-Duong Vuong}
	\affil{Stanford University, \url{{anari,yangpliu,tdvuong}@stanford.edu}}
}

\Tag<ieeetran>{
	\IEEEoverridecommandlockouts
	\title{Optimal Sublinear Sampling of Spanning Trees and Determinantal Point Processes via Average-Case Entropic Independence\thanks{\funding}}
	\author{
		\IEEEauthorblockN{Nima Anari}
		\IEEEauthorblockA{
			Stanford University\\
			anari@cs.stanford.edu
		}
		\and
		\IEEEauthorblockN{Yang P.\ Liu}
		\IEEEauthorblockA{
			Stanford University\\
			yangpliu@stanford.edu
		}
		\and
		\IEEEauthorblockN{Thuy-Duong Vuong}
		\IEEEauthorblockA{
			Stanford University\\
			tdvuong@stanford.edu
		}
	}
}

\date{}

\begin{document}
	\maketitle
	\begin{abstract}
       	We design fast algorithms for repeatedly sampling from strongly Rayleigh distributions, which include as special cases random spanning tree distributions and determinantal point processes. For a graph $G=(V, E)$, we show how to approximately sample uniformly random spanning trees from $G$ in $\widetilde{O}(\lvert V\rvert)$\footnote{Throughout, $\widetilde{O}(\cdot)$ hides polylogarithmic factors in $n$.} time per sample after an initial $\widetilde{O}(\lvert E\rvert)$ time preprocessing. This is the first nearly-linear runtime in the output size, which is clearly optimal. For a determinantal point process on $k$-sized subsets of a ground set of $n$ elements, defined via an $n\times n$ kernel matrix, we show how to approximately sample in $\widetilde{O}(k^\omega)$ time after an initial $\widetilde{O}(nk^{\omega-1})$ time preprocessing, where $\omega<2.372864$ is the matrix multiplication exponent. The time to compute just the weight of the output set is simply $\simeq k^\omega$, a natural barrier that suggests our runtime might be optimal for determinantal point processes as well. As a corollary, we even improve the state of the art for obtaining a single sample from a determinantal point process, from the prior runtime of $\widetilde{O}(\min\{nk^2, n^\omega\})$ to $\widetilde{O}(nk^{\omega-1})$.
       	
       	In our main technical result, we achieve the optimal limit on domain sparsification for strongly Rayleigh distributions. In domain sparsification, sampling from a distribution $\mu$ on $\binom{[n]}{k}$ is reduced to sampling from related distributions on $\binom{[t]}{k}$ for $t\ll n$. We show that for strongly Rayleigh distributions, the domain size can be reduced to nearly linear in the output size $t=\widetilde{O}(k)$, improving the state of the art from $t= \widetilde{O}(k^2)$ for general strongly Rayleigh distributions and the more specialized $t=\widetilde{O}(k^{1.5})$ for spanning tree distributions. Our reduction involves sampling from $\widetilde{O}(1)$ domain-sparsified distributions, all of which can be produced efficiently assuming approximate overestimates for marginals of $\mu$ are known and stored in a convenient data structure. Having access to marginals is the discrete analog of having access to the mean and covariance of a continuous distribution, or equivalently knowing ``isotropy'' for the distribution, the key behind optimal samplers in the continuous setting based on the famous Kannan-Lov\'asz-Simonovits (KLS) conjecture. We view our result as analogous in spirit to the KLS conjecture and its consequences for sampling, but rather for discrete strongly Rayleigh measures.
    \end{abstract}
    \Tag<ieeetran>{
    }
    \Tag{\clearpage}
    \section{Introduction}
\label{sec:intro}

\Tag{Efficiently sampling from probability distributions is a fundamental algorithmic question whose study has been instrumental in revealing connections between many areas of mathematics and computer science. Markov chains are perhaps the single most utilized method in designing sampling algorithms. The study of Markov chains is an active area of research in both high-dimensional continuous settings and combinatorial/discrete settings \cite[see, e.g.,][]{Jer98}. Unlike many other computational tasks, sampling is not in general ``efficiently verifiable.'' This motivates a sharp theoretical understanding of the mixing time of Markov chains, because there is no general technique for knowing when to stop running Markov chains in practice without an a priori theoretical bound.}

In this work, we study how far we can push the runtime of sampling algorithms for the widely used class of strongly Rayleigh distributions \cite{BBL09}, which are distributions supported on size $k$ subsets of a ground set $[n]=\set{1,\dots,n}$, denoted from here on as $\binom{[n]}{k}$, which satisfy strong forms of negative dependence (see \cref{prelim:sr} for a formal definition). Examples of strongly Rayleigh distributions include uniformly random spanning trees in a graph (where $n$ is the number of edges and $k+1$ is the number of vertices) and determinantal point processes.

Prior works \cite{Der19,DCV19,GKMV19,AD20,CDV20,ADVY21} discovered that under certain regularity assumptions on the distribution $\mu$, one can sample from $\mu$ in sublinear ($\ll n$) time. Regularity assumptions are needed to prevent a scenario where an element $i\in [n]$ has an extremely high marginal $\P_{S\sim \mu}{i\in S}$; it is impossible to find out which element has this property (and output it as part of the sample) without examining roughly all the $n$ elements. This is quite reminiscent of the problem of sampling from continuous log-concave densities on the Euclidean space, as was noted in prior works \cite{AD20}, where important directions in the space might be hard to find. The fastest algorithms for sampling from log-concave densities generally proceed by transforming the distribution into an ``isotropic form'' (a time-consuming part of the algorithm) which guarantees no particular direction accounts for a significant part of the variance, and proceed to obtain samples from isotropic log-concave densities \cite{LV18,Chen21,KL22}. The Kannan-Lov\'asz-Simonovits (KLS) conjecture was formulated to answer how fast one can sample from isotropic log-concave densities \cite{LV18}.

Motivated by the analogy with continuous distributions, \textcite{AD20} defined a notion of isotropy for discrete distributions $\mu$ on $\binom{[n]}{k}$: $\mu$ is isotropic when $\P_{S\sim \mu}{i\in S}$ is the same for all $i\in [n]$. A distribution can be \emph{put in approximate isotropic position} via preprocessing (see \cref{sec:isotropic-transform} for details). The main question then becomes
\begin{quote}
    How fast can we sample from (approximately) isotropic distributions $\mu$ on $\binom{[n]}{k}$?
\end{quote}
Prior works \cite{Der19,DCV19,AD20,CDV20} showed that the answer to this is $\leq \poly(k, \log n)$ for strongly Rayleigh distributions, assuming oracle access to $\mu$. However, the optimal sampling runtime remained open. Our main result in this work shows that the optimal runtime for sampling from isotropic strongly Rayleigh distributions on $\binom{[n]}{k}$ is, roughly speaking, at most the runtime for sampling from related distributions on $\binom{[t]}{k}$ for $t=\Otilde{k}$. In other words, isotropy allows us to pretend that $n$ is only as large as $\Otilde{k}$.
\begin{theorem}[Informal, see \cref{thm:overestimate} for a formal statement]
    Suppose that the time to sample from a class of strongly Rayleigh distributions on $\binom{[n]}{k}$ is $\mathcal{T}(n, k)$. Then we can sample from (approximately) isotropic distributions in this family in time $\Otilde{1}\cdot \mathcal{T}(\Otilde{k}, k)$. 
\end{theorem}

\Tag{
\begin{remark}
Our progress is analogous to the history of sampling algorithms for continuous distributions and the role of (continuous) isotropy \cite[see, e.g.,][]{LV18}. Transforming a convex body or a log-concave density into isotropic position (defined as having covariance matrix $\simeq I$ instead of uniform marginals) is the standard preprocessing step, and the main challenge has been establishing properties of isotropic distributions that would then yield optimal bounds on mixing time of standard off-the-shelf Markov chains. A notable conjectured property of isotropic log-concave distributions is the KLS conjecture, which was recently nearly resolved \cite{Chen21,KL22}. We view our result as an analog, at least in spirit, of the KLS conjecture for discrete distributions: we establish optimal mixing time bounds (analogous to  consequences of the KLS conjecture) for strongly Rayleigh (analogous to log-concave) discrete distributions in discrete-isotropic position (a natural analog of the continuous isotropic position). Interestingly, our proof technique also shows some resemblance to the common framework used in recent advances on the KLS conjecture \cite{LV18,Chen21,KL22}: a key technical result we prove is that isotropy is approximately preserved with high probability under a natural localization process (see \cref{sec:concentration}). \Textcite{CE22} observed recently that several localization processes used for continuous and discrete sampling problems can be, at least partially, unified under a single umbrella. We believe our results provide further justification for this unification.
\end{remark}
}

\Tag{
In many applications of sampling, one needs not just one, but rather many independent samples from a distribution. A fundamental observation is that the amortized time of producing many samples can often be much less than the cost of producing one sample. As an example, consider the task of producing samples from a distribution on $n$ points given explicitly by $n$ numbers $p_1,\dots,p_n\geq 0$ that sum to $1$. The time to produce a single sample from this distribution is $\simeq n$, as one needs to look at all $p_i$. However, after reading through the whole input, it is easy to construct a data structure (such as a simple array of prefix sums) that allows subsequent samples to be obtained in $\Otilde{1}$ time. Obtaining similar economies of scale for distributions supported on exponentially-sized state spaces is not possible with this na\"{i}ve approach; instead, our results show how to obtain optimal economies of scale by a different method that preprocesses a strongly Rayleigh distribution and puts it into isotropic form.

We remark that in some natural scenarios, a preprocessing step might not be needed at all, and we can enjoy fast runtimes even for the first sample. For example, if the distribution is symmetric w.r.t.\ the ground set, see, e.g., \cite{OR18} for examples of determinantal point processes on symmetric spaces, the distribution is automatically in isotropic form. Similarly, for random spanning trees in graphs, under mild expansion assumptions (roughly speaking, expanding mildly better than 2-dimensional grids) \cite{AALO18}, no edge will have a large marginal and the distribution is automatically in approximately isotropic form.} Below we expand on two classes of distributions that constitute the main applications of our result.

\paragraph{Random spanning trees.} Random spanning trees of a graph $G=(V, E)$ have found many applications in theoretical computer science. In approximation algorithms for the Traveling Salesperson Problem (TSP) \cite{GW17} they are a key component of the Best-of-Many Christofides algorithm used in recent TSP improvements \cite{KKO21}. Random spanning trees have found applications in the construction of graph sparsifiers \cite{GRV09, KS18}. As another example application, the recent breakthrough of \textcite{KKGZ21} on the $k$-edge connected multi-subgraph problem uses $\Theta(k)$ independent samples of random spanning trees, which demonstrates how economies of scale for sampling can lead to faster algorithms for some natural problems. The distribution of random spanning trees is also deeply connected with spectral graph theory and Laplacians of graphs, e.g., through the matrix-tree theorem. This has all motivated a long sequence of works on obtaining fast algorithms for sampling from this ubiquitous distribution \cite{Aldous90,Broder89,Wil96,CMN96,KM09,MST14,DKPRS17,DPPR17,Sch18,ALOVV21}. Many of these works have used random spanning trees as a testing ground for novel algorithm design techniques, in particular techniques originating in the study of Laplacian solvers, and more recently high-dimensional expanders. The latest works on sampling from spanning trees \cite{Sch18,ALOVV21} obtained, using two very different approaches, nearly-linear time sampling algorithms. In this work we show how to push even further and get optimal sublinear sampling algorithms with runtime $\Otilde{\card{V}}$, after an $\Otilde{\card{E}}$ preprocessing step.

\paragraph{Determinantal point processes.} Another important example of strongly Rayleigh distributions is a $k$-sized determinantal point process, or $k$-DPP for short. A $k$-DPP $\mu$ is a distribution on $\binom{[n]}{k}$ defined with the help of an $n\times n$ positive semidefinite matrix $L$, where probabilities are given by $k\times k$ principal minors:
\[ \mu(S)\propto \det(L_{S, S}). \]
DPPs have found many applications in machine learning, recommender systems, and randomized linear algebra \cite[see, e.g.,][]{DR10,KT12,DM21}. In most applications of $k$-DPPs, the size of a sample is small compared to the domain $[n]$, i.e., $k\ll n$, and the primary goal of sampling algorithms is to minimize the runtime's dependence on $n$. A nearly-linear dependence on $n$ can be achieved for example via Markov chains \cite{AOR16,HS19}. Recent works have shown how to go even further, and obtain after a preprocessing step, $\poly(k, \log n)$ sampling times \cite{DWH18,DWH19,DCV19,GKMV19,Der19,AD20,CDV20}; however, the dependence on $k$ remained suboptimal. We push the runtime to what we believe is the natural barrier for this sampling problem, and obtain a sampling algorithm with runtime $\Otilde{k^\omega}$, where $\omega$ is the matrix multiplication exponent. Note that $k^\omega$ is the time needed to just compute $\mu(S)$ for one set $S$, which is a natural barrier and suggests our result might be optimal.

We further show that the preprocessing step for DPPs can be implemented in time $\Otilde{nk^{\omega-1}}$. This, surprisingly, leads to an improvement for obtaining even a single sample from DPPs. The best prior algorithms were either based on MCMC and had a runtime of $\Otilde{nk^2}$ \cite{AOR16,HS19} or were based on linear algebraic primitives \cite{KT12,Poul20}, which implemented with fast matrix multiplication, would take time $\Otilde{n^\omega}$, see \cref{lem:dpp-matmult}. We remark that our improvement from $\Otilde{\min\set{nk^2,n^\omega}}$ to $\Otilde{nk^{\omega-1}}$ is only made possible by a fast preprocessing step which crucially is implemented by bootstrapping with the primitive of fast sampling from isotropic distributions.

\subsection{Sampling algorithm}

To obtain our optimal sublinear samplers, we use the framework established in prior works \cite{DKPRS17,DCV19,Der19,AD20,CDV20,ADVY21} of sparsifying the domain $[n]$ for isotropic distributions, i.e., distributions with roughly balanced marginals \cite{AD20}. The preprocessing step for our algorithm consists only of putting the distribution into (approximately) isotropic position (see \cref{sec:isotropic-transform}) by finding approximate overestimates for the marginals $\P_{S\sim \mu}{i\in S}$ and transforming $\mu$ by splitting elements with large marginals. One of our novel contributions is the design of new schemes for bootstrapping very fast (and likely optimal) preprocessing steps.

For our main contribution, we obtain an optimal nearly-linear-in-$k$ domain sparsification for isotropic strongly Rayeligh distributions. In domain sparsification, we reduce the task of sampling from our distribution on $\binom{[n]}{k}$ to distributions on $\binom{[t]}{k}$; we show this can be done with $t=\Otilde{k}$. Prior works on this problem either required $t\simeq k^2$ \cite{DCV19,Der19,AD20,CDV20,ADVY21} or for the specific case of spanning tree distributions required $t\simeq k^{1.5}$ \cite{DKPRS17}.

More formally, for an approximately isotropic $\mu$, we generate a sample by starting from some set $S_0\in \binom{[n]}{k}$ and following the random walk defined by \cref{alg:random-walk} for $\Otilde{1}$ steps. We output $S_{\Otilde{1}}$ as our approximate sample from $\mu$. Note that this random walk has an easy step (choosing $T_i$ uniformly at random from supersets of $S_i$) and a challenging step (choosing $S_{i+1}$ from subsets of $T_i$ with law induced by $\mu$). The challenging step is an instance of a similar sampling problem but with a smaller domain size $t$, so we can use a problem-specific baseline sampling algorithm.

\begin{Algorithm}
\caption{Down-up walk on the complement distribution\label{alg:random-walk}}
\For{$i=0,1,2,\dots$}{
From all $t$-sized supersets of $S_i$, select one uniformly at random and name it $T_i$.\;
Select among $k$-sized subsets of $T_i$ a random set $S_{i+1}$ with $\P{S_{i+1}}\propto\mu(S_i)$.
}
\end{Algorithm}

\Tag{We remark that the sparsification algorithm (\cref{alg:random-walk}) is not new and very similar variants of it have been used by almost all mentioned prior works. However, our analysis of \cref{alg:random-walk} is entirely different. A departure from prior methods of analysis is not for convenience, but rather necessary. Domain sparsification looks fundamentally different below $t\simeq k^{1.5}$. All prior works used in some shape or form the fact that the partition function of $T_i$, i.e. $\sum_{S\subseteq T_i} \mu(S)$ concentrates for a random $T_i$. Indeed, \textcite{DKPRS17} used this to design algorithms for not just sampling, but also counting spanning trees. Below the threshold of $t\simeq k^{1.5}$, the partition function no longer concentrates (see \cref{sec:comparison}). Surprisingly, we still show that while $T_i$'s are not good representatives of the ground set $[n]$ for partition functions or counting purposes, they still are good sparsifiers for sampling.}

\subsection{Our results}
To formally state our results on sampling from strongly Rayleigh distributions, it is useful to define $\mathcal{T}_{\mu}(t, k)$ for a distribution $\mu \in \R^{\binom{[n]}{k}}$ as the time it takes to produce a sample from $\mu$ conditional on all elements of the sample being a subset of a fixed set $T$ of size $\card{T} = t$. We use $\Otilde{\cdot}$ to suppress $\poly\log n$ factors. Notice below that the sum of marginals $\sum_i \P_{S\sim \mu}{i\in S}$ is always equal to $k$ for a distribution supported on $\binom{[n]}{k}$.

\begin{theorem}[Sampling using marginal overestimates]
\label{thm:overestimate}
Given a strongly Rayleigh distribution $\mu \in \R^{\binom{[n]}{k}}$ and marginal overestimates $q_i \ge \P_{T \sim \mu}{i \in T}$ for $i \in [n]$ which sum to $K := \sum_{i\in[n]} q_i,$ there is an algorithm that produces a sample from a distribution with total variation distance $n^{-O(1)}$ from $\mu$ in time bounded by $\Otilde{1}$ calls to $\mathcal{T}_{\mu}(O(K), k).$
\end{theorem}
We prove \cref{thm:overestimate} using a local-to-global argument, which requires us to also show that random conditionals of $\mu$ are isotropic with high probability. This is similar in spirit to the recent analyses of the KLS conjecture using stochastic localization \cite{Chen21,KL22} which show that an isotropic continuous distribution remains approximately isotropic over a stochastic evolution.
\begin{theorem}[Informal, see \cref{thm:concentration} for a formal statement]
\label{thm:informalconcentration}
Let $\mu \in \R^{\binom{[n]}{k}}$ be an isotropic strongly Rayleigh distribution. For $T \subseteq [n]$ and $S \in \binom{T}{k}$, let $\mu_T(S) := \mu(S)/\sum_{S \in \binom{T}{k}} \mu(S)$. Then with high probability over $T \in \binom{[n]}{t}$ for $t = \Otilde{k}$, $\mu_T$ is approximately isotropic.
\end{theorem}
Our input distributions may not be isotropic, so we also design an efficient preprocessing step to obtain marginal estimates to transform $\mu$ into an isotropic distribution.
\begin{theorem}[Informal, see \cref{lemma:marginalestimate} for a formal statement]
\label{thm:marginalinformal}
Given access to a strongly Rayleigh distribution $\mu \in \R^{\binom{[n]}{k}}$, we can obtain overestimates of the marginals $\P_{T \sim \mu}{i \in T}$ summing to $O(k)$ in time proportional to $\Otilde{n/k}$ calls to a sampler for isotropic distributions on sets of size $\Otilde{k}$.
\end{theorem}
We remark that the preprocessing time of $\Otilde{\card{E}}$ for estimating marginals of a random spanning tree can be alternatively achieved by estimating effective resistances of the graph using Laplacian solvers and the Johnson-Lindenstraus lemma \cite{ST04,SS11}. However, we give a self-contained method by bootstrapping the sampling algorithm (\cref{thm:marginalinformal}) that covers not only random spanning trees, but also $k$-DPPs.

We can apply these results along with known algorithms that sample a random spanning tree in $\Otilde{\card{E}}$ time \cite{ALOVV21}, or a $k$-DPP on $n$ elements in $\Otilde{n^\omega}$ time (see \cref{lem:dpp-matmult}) to achieve faster runtimes for sampling from these distributions. We note that our algorithm for sampling a $k$-DPP is faster than previously known runtimes, even in the case of generating a single sample.
\begin{corollary}[Sampling spanning trees]
\label{cor:spanning}
For a graph $G = (V, E)$, possibly weighted with weights $\lambda \in \R_{>0}^E$, we can output $s$ independent spanning trees with $n^{-O(1)}$ total variation distance from the distribution $\mu(T) \propto \prod_{e \in T}\lambda_e$ in time $\Otilde{\card{E} + s\card{V}}$.
\end{corollary}

\begin{corollary}[Sampling DPPs]
\label{cor:dpp}
Given an $n \times n$ positive semidefinite matrix $L$, there is an algorithm that outputs $s$ independent approximate samples from the $k$-DPP defined by $L$ in time $\Otilde{nk^{\omega-1} + sk^\omega}.$
\end{corollary}

Finally, we remark that our methods also show analogous mixing times of $\Otilde{k}$ steps for the Markov chain that uses small up-down steps, i.e., calls to $\mathcal{T}_{\mu}(k+1, k)$, when sampling isotropic strongly Rayleigh distributions. Such steps are easy to implement in practice, and were used to efficiently sample from general strongly Rayleigh and logconcave distributions \cite{CGM19,ALOVV21}. We formally state these results in \cref{thm:updown} in \cref{sec:mainproof}.

\subsection{Techniques and comparison to prior work}\label{sec:comparison}
We depart from previous analyses of \cref{alg:random-walk} and take the new approach of viewing the sparsification procedure as a down-up random walk on high-dimensional expanders \cite[see, e.g.,][]{KO18}. We establish that isotropy significantly improves the ``expansion'' of the high-dimensional-expander. We use the notion of expansion called entropic independence \cite{AJKPV21}, which is one of the few able to yield modified log-Sobolev inequalities and tight mixing times for down-up walks.

The random walk in \cref{alg:random-walk} can be seen as the down-up walk (see \cref{sec:prelim} for definition) on the complement/dual distribution associated with $\mu$; note that each step of this walk involves a sparsified sampling problem, where we only have to sample from a distribution on $\binom{T_i}{k}$. For this we use a baseline sampling algorithm, a Markov chain based on a clever link-cut tree data structure for spanning trees, and a na\"{i}ve matrix-multiplication-based sampler for DPPs.

Below we describe the main techniques we use.

\paragraph{Boosted entropic independence under isotropy.} The main tool we use to bound the mixing time of the random walk is the notion of entropic independence (see \cref{sec:prelim} for definition) \cite{AJKPV21}. While standard results about strongly Rayleigh distributions give an out-of-the-box factor $1$ entropic independence for the complement distribution $\bmu$, this is too weak for our purposes as it only implies a mixing time of $\simeq \Otilde{(n-k)/t}$ for \cref{alg:random-walk}, which has an unacceptable dependence on $n$. This is not surprising, as these black-box results do not incorporate isotropy of $\mu$. In this work, we show that whenever $\mu$ has entropic independence and its marginals are not too large, the complement distribution $\bmu$ has to have a boosted entropic independence, better by a logarithmic factor over what is na\"{i}vely expected (\cref{thm:one level contraction}).

\paragraph{Average case local-to-global and concentration of marginals.}
The standard machinery for establishing mixing times using entropic independence (i.e., the so-called local-to-global method \cite{AL20}) needs entropic independence of not just the distribution $\bmu$, but all of its conditionings as well. Conditioning $\bmu$ on a set of elements is the same as throwing those elements out of the ground set for $\mu$. Unfortunately, in the worst case, this can significantly imbalance the marginals of $\mu$. \Tag{As an example, consider the spanning tree distribution on a complete graph, which is by symmetry isotropic. Throwing edges out, we can create any graph as a subgraph of the complete graph; for example, we can throw out all but one edge in a cut to make the marginal of that edge equal to $1$.} To overcome this obstacle, we show that with high probability, i.e., in an average sense over the choice of elements in the conditioning, the marginals remain balanced (\cref{thm:concentration}) and combine this with an average local-to-global result adapted from \cite{AASV21} (\cref{thm:entropy contraction}) to establish the tight mixing time. As far as we know, this is the first application of an average local-to-global theorem. Our strategy of showing average-case isotropy under conditionings is reminiscent of the strategy employed in works on the KLS conjecture which show approximate isotropy holds under an appropriate localization process \cite{Chen21,CE22,KL22}.

\paragraph{Improved marginal estimation.} Our main focus is on the time per sample after preprocessing, but we also obtain fast algorithms that improve the preprocessing runtime compared to prior works. Our improved procedures are able to shave off $\poly(k)$ factors from the runtime of marginal estimation (\cref{lemma:marginalestimate}), and are essential for our faster $\Otilde{nk^{\omega-1}}$ time algorithm for sampling from a $k$-DPP. This is achieved by a recursive procedure that uses marginals of the restriction of $\mu$ to roughly half the domain $[n]$ as overestimates for the marginals of $\mu$. In the end, marginal overestimation is roughly reduced to $\simeq \Otilde{n/k}$ subtasks of marginal overestimation for distributions over domains of size $\Otilde{k}$.

\paragraph{Barriers faced by prior approaches.} In order to derive the tight sparsification of $t=\Otilde{k}$ in \cref{alg:random-walk}, we had to rethink the entire analysis technique. To emphasize the importance of tight bounds on $t$, we note that prior results on general strongly Rayleigh measures \cite{Der19,CDV20,DCV19,AD20} had at least a quadratic dependence on the output size $k$, which made them moot for random spanning trees (where $k^2$ is always larger than the total number of edges in the graph). The barrier faced by the aforementioned works, and also that of \cite{DKPRS17} is roughly speaking that for the regime $t=\Otilde{k}$, subsets $T_i$ are \emph{not} good sparsifiers for partition functions. To appreciate this better, consider a simple distribution $\mu$ on $\binom{[n]}{k}$ defined as follows: first we partition $[n]$ into disjoint sets $U_1,\dots,U_k$ of size $n/k$ each, and then define our distribution as uniform over sets which pick exactly one element from each $U_i$. Clearly this distribution is isotropic. Now suppose that we select a uniformly random $ck$-sized set $T$ from $[n]$. The intersection of $T$ with each $U_i$ has expected size $c$. For small values of $c$, the distribution of this intersection size is well-approximated by a Poisson distribution. The count / partition function of the distribution restricted to $T$ is $\prod_{i=1}^k \card{T\cap U_i}$.

The fluctuations of each $T\cap U_i$ are on the order of $\sqrt{c}$. These fluctuations make the above product typically very far from its mean, unless $c$ is growing at least polynomially with $k$. A careful analysis (similar to \cite{DKPRS17}) would show that $c\simeq \sqrt{k}$ is the threshold after which the count concentrates around the mean. To overcome this barrier, we do not use counts in our analysis at all. Rather, we show that \emph{marginals} do concentrate all the way down to the threshold $t=\Otilde{k}$, using a martingale argument. We combine this concentration of marginals with the fact that isotropy improves entropic independence to show that isotropic strongly Rayleigh distributions are extremely good high-dimensional expanders in an average sense.

\Tag{
\subsection{Organization} In \cref{sec:prelim} we collect preliminary notions relating to distributions, conditionals, and Markov chains. We additionally introduce entropic-independence and local-to-global theorems that we use to analyze the down-up walk that our sampling algorithms are based on. In \cref{sec:entropy} we show our main bound on the entropy contraction of a down step of the complement distribution of a strongly Rayleigh distribution with bounded marginals. In \cref{sec:concentration} we show that random marginals of strongly Rayleigh distributions stay bounded with high probability, which is essential to applying the average-case local-to-global principle. In \cref{sec:estimation} we give a simple and efficient procedure for estimating marginal overestimates based on recursive sampling. In \cref{sec:mainproof} we combine the previous sections to prove our main results about sampling spanning trees, DPPs, and strongly Rayleigh distributions in general. Finally, deferred proofs are given in \cref{sec:defer}.
}
    \Tag{
		\subsection*{Acknowledgments}
		\acks{}
		
		\funding{}
	}
    \section{Preliminaries}
\label{sec:prelim}
We use $[n]$ to denote the set $\set{1,\dots,n}$. We view distributions/measures defined over a finite ground set $\Omega$ interchangeably as either (probability mass) functions $\mu:\Omega \to \R_{\geq 0}$ or just row vectors $\mu\in \R^\Omega$.

For a distribution $\mu\in \R^{\binom{[n]}{k}}$, let $p(\mu) \in \R^n$ denote the marginals of $\mu$, i.e., $p(\mu)_i := \P_{S \sim \mu}{i \in S}.$ Denote $p(\mu)^{\max} := \max\set{p(\mu)_i\given i\in [n]}.$ When $\mu$ is clear from context, we write $p$ instead of $p(\mu).$ We define $\bmu: \binom{[n]}{n-k}\to \R_{\geq 0}$ as the \emph{complement distribution} associated to $\mu$, defined as
\[ \bmu(S) := \mu([n]\backslash S).\]

Our analysis (in particular for applying a local-to-global principle) requires looking at restrictions of $\mu$ to specific subset of the ground set $[n]$ of elements. In the complement, this corresponds to conditioning that $\bar{\mu}$ contains certain elements.
\begin{definition}[Restricted distribution]
\label{def:conditional}
For a distribution $\mu$ defined over subsets of a ground set $[n]$ and $S \subseteq [n]$, define $\mu_S$ to be the distribution of $F \sim \mu$ restricted to the set $S$, i.e., conditioned on the event $F \subseteq S.$
\end{definition}
\begin{definition}[Conditional distribution]\label{def:conditioned}
For a distribution $\mu$ defined over subsets of a ground set $[n]$ and $T \subseteq [n]$, define $\mu^T$ to be the distribution of $F \sim \mu$ conditioned on the event $F \supseteq T.$
\end{definition}

\subsection{Markov chains and functional inequalities}
\label{prelim:markov}
Let $\mu$ and $\nu$ be probability measures on a finite set $\Omega$. The Kullback-Liebler divergence (or relative entropy) between $\nu$ and $\mu$ is given by
\[\DKL{\nu \river \mu} = \sum_{x \in \Omega}\nu(x)\log\parens*{\frac{\nu(x)}{\mu(x)}},\]
with the convention that this is $\infty$ if $\nu$ is not absolutely continuous with respect to $\mu$. \Tag{By Jensen's inequality, $\DKL{\nu \river \mu} \geq 0$ for any probability measures $\mu, \nu$.} The total variation distance between $\mu$ and $\nu$ is given by
\[d_\TV(\mu, \nu) = \frac{1}{2}\sum_{x \in \Omega}\abs{\mu(x) - \nu(x)}.\]

A Markov chain on $\Omega$ is specified by a row-stochastic non-negative transition matrix $P \in \R^{\Omega \times \Omega}$. We refer the reader to \cite{LP17} for a detailed introduction to the analysis of Markov chains. As is common, we will view probability distributions on $\Omega$ as row vectors. Recall that a transition matrix $P$ is said to be reversible with respect to a distribution $\mu$ if for all $x,y \in \Omega$, $\mu(x)P(x,y) = \mu(y)P(y,x)$. In this case, it follows immediately that $\mu$ is a stationary distribution for $P$, i.e., $\mu P = \mu$. If $P$ is further assumed to be ergodic, then $\mu$ is its unique stationary distribution, and for any probability distribution $\nu$ on $\Omega$, $d_\TV(\nu P^{t}, \mu) \to 0$ as $t \to \infty$. The goal of this paper is to investigate the rate of this convergence. 

\begin{definition}
Let $P$ be an ergodic Markov chain on a finite state space $\Omega$ and let $\mu$ denote its (unique) stationary distribution. For any probability distribution $\nu$ on $\Omega$ and $\epsilon \in (0,1)$, we define $t_\mix(P, \nu, \epsilon)$ to be
\[ \min\set{t\geq 0 \given d_\TV(\nu P^{t}, \mu)\leq \epsilon},\]
and let $t_\mix(P,\epsilon)$ denote
\[\max\set*{\min\set{t\geq 0 \given d_\TV(\1_{x}P^t, \mu) \leq \epsilon}\given x\in \Omega},\]
where $\1_{x}$ is the point mass supported at $x$. 
\end{definition}

We will drop $P$ and $\nu$ if they are clear from context. Moreover, if we do not specify $\epsilon$, then it is set to $1/4$. This is because the growth of $t_{\operatorname{mix}}(P,\epsilon)$ is at most logarithmic in $1/\epsilon$ \cite[cf.][]{LP17}. 

\begin{lemma}
\label{lem:entropy-contraction-implies-mlsi}
Let $\mu$ be a probability measure on the finite set $\Omega$. Let $P$ denote the transition matrix of an ergodic, reversible Markov chain on $\Omega$ with stationary distribution $\mu$. Suppose there exists some $\alpha \in (0,1]$ such that for all probability measures $\nu$ on $\Omega$, we have
\[\DKL{\nu P \river \mu P} \leq (1-\alpha)\DKL{\nu \river \mu}.\]
Then $t_\mix(P,\epsilon)\leq$
\[ \ceil*{ \frac{1}{\alpha} \cdot \parens*{\log\log\parens*{\frac{1}{\min\set{\mu(x)\given x\in \Omega}}} + \log\parens*{\frac{1}{2\epsilon^2}}} }.\]
\end{lemma}
This is the standard argument for bounding mixing times via modified log-Sobolev inequalities and can be found in, e.g., \cite{BT06}.

\subsection{Strongly Rayleigh distributions}
\label{prelim:sr}
For density function $\mu:\binom{[n]}{k}\to\R_{\geq 0}$, the generating polynomial of $\mu$ is the multivariate $k$-homogeneous polynomial defined as follows:
\[ g_\mu(z_1,\dots,z_n)=\sum_{S\in \binom{[n]}{k}} \mu(S)\prod_{i\in S} z_i. \]
\begin{definition}
    Consider the open half-plane $H = \set*{z \given \Im(z) > 0} \subseteq \C$. We say a polynomial $g(z_1,\cdots, z_n) \in \R[z_1, \cdots, z_n]$ is real-stable if $g$ does not have roots in $H^n.$ For convenience, the zero polynomial is taken to be real-stable.
\end{definition}
A distribution $\mu: 2^{[n]} \to \R_{\geq 0}$ is strongly Rayleigh iff its generating polynomial is real stable \cite[see][]{BBL09}. If $\mu$ is strongly Rayleigh, then its conditional and restricted distributions (see \cref{def:conditional,def:conditioned}) are also strongly Rayleigh. The key fact we use about strongly Rayleigh distributions is that they are negatively correlated \cite[see, e.g.,][]{BBL09}, i.e., the marginals (of non-restricted elements) increase under restrictions (\cref{def:conditional}): 
\[ \P_{S\sim \mu}{i\in S}\leq \P_{S\sim \mu_T}{i\in S}\enspace\text{for}\enspace i\notin T. \]

\subsection{Down-up and up-down walks}
\begin{definition}[Down operator]
\label{def:down}
	For $\l\leq k$ define the row-stochastic matrix $D_{k\to \l}\in \R_{\geq 0}^{\binom{[n]}{k}\times \binom{[n]}{\l}}$ by
	\[ D_{k\to \l}(S, T)=\begin{cases}
		0 & \text{if }T\not\subseteq S,\\ 
		\frac{1}{\binom{k}{\l}}& \text{otherwise}.
	\end{cases}\]
	Note that for a distribution $\mu$ on size $k$ sets, $\mu D_{k\to \l}$ will be a distribution on size $\l$ sets. In particular, $\mu D_{k\to 1}$ will be the vector of normalized marginals of $\mu$: $(\P{i\in S}/k)_{i\in [n]}$, i.e., $p(\mu)/k$.
\end{definition}

\begin{definition}[Up operator]
\label{def:up}
For $\ell \le k$ define the row-stochastic matrix $U_{\ell\to k} \in \R_{\geq 0}^{\binom{[n]}{\l}\times \binom{[n]}{k}}$ by
\[
U_{\l\to k}(T, S)=\begin{cases}
		0 & \text{if }T\not\subseteq S,\\ 
		\frac{\mu(S)}{\sum_{T\ni S'}\mu(S')} & \text{otherwise}.
	\end{cases}
\]
\end{definition}

As in \cite{AD20,ADVY21}, we consider the following Markov chain $M^t_{\mu}$ defined for any positive integer
$t$, with the state space $\supp(\mu).$ Starting from $S_0 \in \supp(\mu)$, one step of the chain is:

\begin{enumerate}
    \item Sample $T \in \binom{[n] \setminus S_0}{t - k}$ uniformly at random.
    
    \item Downsample $S_1 \sim \mu_{S_0 \cup T}$, where $\mu_{S_0 \cup T}$ is $\mu$ restricted to $S_0 \cup T$\Tag<comment>{, a.k.a.\ $\1_{S_0\cup T}\star \mu$,} and update $S_0$ to be $S_1.$
\end{enumerate}

\begin{proposition}
The complement of $S_1$ is distributed according to $ \bmu_0 D_{(n-k) \to (n-t)} U_{(n-t) \to (n-k)} $ where $\mu_0$ is the distribution of the set $S_0.$
\end{proposition}

\begin{proposition} \label{prop:intermediate sampling detail balance}
For any distribution $\mu: \binom{[n]}{k}\to \R_{\geq 0}$ that is strongly Rayleigh, the chain $M^t_{\mu}$ for $t \geq k+1$ is irreducible, aperiodic and has stationary distribution $\mu.$
\end{proposition}
\subsection{Entropic independence}
We say that a distribution $\mu$ is entropically independent if the down operator $D_{k\to1}$ significantly contracts the relative entropy between $\nu$ and $\mu$ for any distribution $\nu$.
\begin{definition}[Entropic independence \cite{AJKPV21}]
    \label{def:entropic-independence}
A probability distribution $\mu$ on $\binom{[n]}{k}$ is said to be $(1/\alpha)$-entropically independent, if for all probability distributions $\nu$ on $\binom{[n]}{k}$,
\[ \DKL{\nu D_{k\to 1} \river \mu D_{k\to 1}}\leq \frac{1}{\alpha k}\DKL{\nu \river \mu}.  \]
\end{definition}
Any distribution with a log-concave generating polynomial (e.g., uniform on bases of a matroid) is $1$-entropically independent. This includes all strongly Rayleigh distributions.
\begin{lemma}[{\cite[Theorem 4]{AJKPV21}}]
\label{lemma:1entropic}
Any strongly Rayleigh $\mu$ is $1$-entropically independent. The conditional and restricted distributions of $\mu$ are also strongly Rayleigh, and thus $1$-entropically independent. 
\end{lemma}
\subsection{Average-case local-to-global method}
First, we define the notion of the link of the distribution $\mu$ w.r.t.\ a set $T$ \cite[see, e.g.,][]{KO18}. This is almost the same as the conditioned distribution $\mu^T$, see \cref{def:conditioned}, except we remove the set $T$.
\begin{definition}
	For a distribution $\mu:\binom{[n]}{k}\to\R_{\geq 0}$ and a set $T\subseteq [n]$ of size at most $k$, we define the \emph{link of $T$} to be the distribution $\mu^{-T}:\binom{[n]-T}{k-\card{T}}\to \R_{\geq 0}$ which describes the law of the set $S-T$ where $S$ is sampled from $\mu$ conditioned on the event $S\supseteq T$.
\end{definition}


We show that entropic independence for links, i.e., contraction of KL-divergence by $D_{k\to 1}$ operators, results in the contraction of KL-divergence by $D_{k\to \l}$ operators for larger $\l$. While this is by now a well-understood phenomenon, sometimes called the local-to-global method \cite[see, e.g.,][]{AL20}, we use an \emph{average case} variant, adapted from \cite{AASV21}, which only requires entropic independence for a ``typical'' link as opposed to a worst case link. \Tag{We use that we can imagine deleting conditioned out elements in a \emph{random order}. This is essential for our result, where it is \emph{not true} that for every permutation the resulting product of $\rho(\cdot)$ below is large enough.}
\begin{theorem}\label{thm:entropy contraction}
	Suppose that for every set $T$ of size $\leq k-2$, $\mu^{-T}$ contracts KL-divergence in terms of factor parameterized by $\rho(T)$:
	\Tag{\[
		\DKL{\nu D_{k-\card{T}\to 1} \river \mu^{-T}D_{k-\card{T}\to 1}}\leq (1-\rho(T))\DKL{\nu\river \mu^{-T}}.
	\]}
	\Tag<ieeetran>{\begin{multline*}
		\DKL{\nu D_{k-\card{T}\to 1} \river \mu^{-T}D_{k-\card{T}\to 1}}\leq\\ (1-\rho(T))\DKL{\nu\river \mu^{-T}}.
	\end{multline*}}
	In other words assume that $\mu^{-T}$ is $(k-\card{T})(1-\rho(T))$-entropically independent. For a set $T$, define the harmonic mean
	\Tag{
		\[ \gamma_T:=\E*_{e_1,\dots,e_{\card{T}}~\text{uniformly random permutation of }T}{\parens*{\rho(\emptyset)\rho(\set{e_1})\rho(\set{e_1,e_2})\cdots \rho(\set{e_1,\dots,e_{\card{T}-1}})}^{-1}}^{-1}. \]
	}
	\Tag<ieeetran>{
		\[ \gamma_T:=\E*{\parens*{\rho(\emptyset)\rho(\set{e_1})\cdots \rho(\set{e_1,\dots,e_{\card{T}-1}})}^{-1}}^{-1}. \]
		where $e_1,\dots,e_{\card{T}}$ is a random permutation of $T$.
	}
	Then the operator $D_{k\to \l}$ has KL-divergence contraction
	\[ \DKL{\nu D_{k\to\l} \river \mu D_{k\to\l}} \leq (1-\kappa) \DKL{\nu \river \mu}, \]
	with
	\[ \kappa:=\min\set*{\gamma_T\given T\in \binom{[n]}{\l}}. \]
\end{theorem}
The proof is similar to \cite[Theorem 46]{AASV21}, \cite[Theorem 5]{AJKPV21}, and is deferred to \Tag{\cref{sec:defer}}\Tag<ieeetran>{the full version}.
\begin{remark}
	Similar to \cite{AASV21}, if the KL-divergence is replaced by any other type of $f$-divergence, a common choice being $\chi^2$-divergence which roughly relates to the notion of spectral independence, \cref{thm:entropy contraction} still remains valid.	
\end{remark}

\subsection{Isotropic transformation}\label{sec:isotropic-transform}
\Textcite{AD20} introduced the following subdivision process that takes marginal overestimates of an arbitrary distribution $\mu \in \R^{\binom{[n]}{k}}$, and transforms sampling from $\mu$ to sampling from a distribution with nearly uniform marginals. In the following, we call $\mu'$ the \emph{isotropic transformation} of $\mu$.
 
\begin{definition} \label{def:isotropic-transformation}
Let $\mu: \binom{n}{k} \to \R_{\geq 0}$ be an arbitrary probability distribution, and assume that for some constant $c \geq 1,$ we have marginal overestimates $p_1,\dots, p_n$ of the marginals with $p_1+\dots+p_n \le K$ and $p_i\geq \P_{S\sim \mu}{i\in S}$ for all $i$.
Let $t_i:=\ceil{\frac{n}{K}p_i}$. We will create a new distribution out of $\mu$: For each $i \in [n]$, create $t_i$ copies of the element $i$ and let the collection of all these copies be the new ground set: $U = \bigcup_{i = 1}^{n} \set{i^{(1)}, \ldots, i^{(t_i)}}$. Define the following distribution $\mu': \binom{U}{k} \to \R_{\geq 0}$ from $\mu$:
\[
\mu'\parens*{\set*{i_1^{(j_1)}, \ldots, i_k^{(j_k)}}}:=\frac{\mu(\set{i_1, \ldots, i_k})}{t_1 \cdots t_k}.
\]
Another way we can think of $\mu'$ is that to produce a sample from it, we can first generate a sample $\set{i_1, \ldots, i_k}$ from $\mu$, and then choose a copy $i_m^{(j_m)}$ for each element $i_m$ uniformly at random. 
\end{definition}
As show in \cite[Proposition 24]{ADVY21}, performing the isotropic transformation in \cref{def:isotropic-transformation} at most only doubles the size of the universe $U$, but makes all marginals bounded by $O(K/n)$ now. \Tag{For the convenience of the reader, we present the proofs of the following in \cref{sec:defer}.}
\begin{proposition} \label{prop:near-isotropic} 
Let $\mu: \binom{n}{k} \to \R_{\geq 0}$, and let $\mu': \binom{U}{k} \to \R_{\geq 0}$ be the subdivided distribution from \cref{def:isotropic-transformation}. The following hold for $\mu'$:
\begin{enumerate}
\item Near-isotropy: For all $i^{(j)} \in U$, the marginal $\P_{S \sim \mu'}{i^{(j)} \in S} \leq \frac{K}{n} \leq \frac{2K}{\abs{U}}$.
\item Linear ground set size: $\card{U} \leq 2n.$
\item If $ \mu$ is strongly Rayleigh then so is $\mu'.$
\end{enumerate}
\end{proposition}
    \section{Entropy contraction}
\label{sec:entropy}
Our goal is to prove an entropy contraction inequality for the $(n-k)\to1$ down operator. For strongly Rayleigh distributions $\bmu \in \R^{\binom{[n]}{k}}$, which are $1$-entropically independent (\cref{def:entropic-independence}), the entropy contracts by $1/(n-k)$. Surprisingly, if $\mu$ also has nearly uniform marginals, the entropy contracts even more, by an extra $\sim \log(n/k)$ factor.

\begin{theorem}[Level one entropy contraction] \label{thm:one level contraction}
Let $\mu \in \R^{\binom{[n]}{k}}$ be a $1$-entropically independent distribution with $p(\mu)^{\max} : = \max_{i\in [n]} \P_{F\sim\mu}{i\in F} \leq \frac{1}{100}$.
Then for any distribution $\bnu \subseteq \R^{\binom{[n]}{n-k}}$,
\Tag{
	\[\DKL{\bnu D_{(n-k)\to1} \river  \bmu D_{(n-k)\to1}} \leq \frac{1}{(n-k) \log ((e p(\mu)^{\max})^{-1})  } \DKL{\bnu \river \bmu}. \]
}
\Tag<ieeetran>{
	\begin{multline*}
		\DKL{\bnu D_{(n-k)\to1} \river  \bmu D_{(n-k)\to1}} \leq\\ \frac{1}{(n-k) \log ((e p(\mu)^{\max})^{-1})  } \DKL{\bnu \river \bmu}.
	\end{multline*}
}
\end{theorem}

We show this theorem by directly comparing the relative entropies of $p = \nu D_{k\to1}$ and $q = \bnu D_{(n-k)\to1}$ with respect to $\mu D_{k\to1}$ and $\bmu D_{(n-k)\to1}$ respectively. We first show a single-variable instance of this, which we sum over to get the overall comparison in \cref{lemma:compare_KL}.
\begin{lemma}
\label{lemma:tangentline}
Let $p, q, \alpha \in \R_{\ge0}$ be such that $\alpha p + (1-\alpha)q = 1$. If $\alpha < 1/100$ then for any $\mu \in (0, 1/100)$,
\[ C_{\alpha,\mu} p \log(\alpha p/\mu) - q \log((1-\alpha)q/(1-\mu)) \ge K_{\alpha,\mu}\parens*{\alpha p - \mu} \]
for any constants $C_{\alpha, \mu} \ge \frac{\alpha}{(1-\alpha)\log(1/(e\mu))}$ and $K_{\alpha, \mu} := C_{\alpha,\mu}/\alpha + (1-\alpha)^{-1}$.
\end{lemma}
We explain some intuition behind this claim. First, both sides vanish when $p = \mu/\alpha$ and $q = (1-\mu)/(1-\alpha)$. The constants $C_{\alpha,\mu}$ and $K_{\alpha,\mu}$ are chosen so that the inequality is tight up to the second order at this point where $p = \mu/\alpha$ and $q = (1-\mu)/(1-\alpha)$.
\begin{proof}[Proof of \cref{lemma:tangentline}]
Define $f(q) := C_{\alpha,\mu} p \log(\alpha p/\mu) - q \log((1-\alpha)q/(1-\mu)) - K_{\alpha,\mu}\parens*{\alpha p - \mu}$. This is defined for $q \in [0, 1/(1-\alpha)]$.
Note that $\frac{d}{dq} p = -\frac{1-\alpha}{\alpha}.$ Hence
\Tag{
	\[ f'(q) = -C_{\alpha,\mu} \cdot \frac{1-\alpha}{\alpha} \parens*{\log(\alpha p/\mu) + 1} - \parens*{\log\parens*{\frac{(1-\alpha)q}{1-\mu}} + 1} + K_{\alpha,\mu}(1-\alpha). \]
}
\Tag<ieeetran>{
	\begin{multline*} f'(q) = -C_{\alpha,\mu} \cdot \frac{1-\alpha}{\alpha} \parens*{\log(\alpha p/\mu) + 1} -\\ \parens*{\log\parens*{\frac{(1-\alpha)q}{1-\mu}} + 1} + K_{\alpha,\mu}(1-\alpha). \end{multline*}
}
By our careful choice of $K_{\alpha,\mu}$, we have $f'(\bar{q}) = 0$ for $\bar{q} = \frac{1-\mu}{1-\alpha}$.
Additionally, we can calculate
\begin{align}
\label{eq:secondderiv}
    f''(q) = p^{-1} q^{-1} \parens*{C_{\alpha,\mu}\parens*{\frac{1-\alpha}{\alpha}}^2 q -  \frac{1 - (1-\alpha) q }{\alpha}}.
\end{align}
Note that the $f''(q) = 0$ at exactly one value of $q$, which we denote by $q_2.$ Observe that $f''(q) \geq 0$ if and only if $q \geq q_2$, because the coefficient of $q$ in \cref{eq:secondderiv} is positive. Hence
\begin{align*} f''(\bar{q}) &= p^{-1} \bar{q}^{-1} \parens*{C_{\alpha,\mu}\parens*{\frac{1-\alpha}{\alpha}}^2 \bar{q} -  \frac{1 - (1-\alpha) \bar{q} }{\alpha}}\\
\Tag{&\geq p^{-1} \bar{q}^{-1}\parens*{ \frac{\alpha}{(1-\alpha)\log(1/(e\mu))}\parens*{\frac{1-\alpha}{\alpha}}^2 \frac{1-\mu}{1-\alpha} -  \frac{1 - (1-\alpha) \frac{1-\mu}{1-\alpha}  }{\alpha}}  \\}
&\Tag{=}\Tag<ieeetran>{\geq}  p^{-1} \bar{q}^{-1}\parens*{ \frac{1}{\log(1/(e\mu))} \cdot \frac{1-\mu}{\alpha} -  \frac{\mu  }{\alpha}}\\
&= p^{-1} \bar{q}^{-1} \frac{1}{\alpha\log(1/(e\mu))} \parens*{1- \mu (1+ \log(1/(e\mu))) }\\
&= p^{-1} \bar{q}^{-1} \frac{1}{\alpha\log(1/(e\mu))} \parens*{1 + \mu \log \mu } \geq 0
\end{align*}
because $\mu \log \mu \ge -1$ for $\mu \le 1/100$.

 Note that $f'(0) = f'(1/(1-\alpha)) = +\infty.$ Recall that $f''(q) < 0$ for any $q < q_2$ and $f''(q) > 0$ for $q > q_2.$ The above calculation implies $\bar{q}\geq q_2.$ 
 Thus $0=f'(\bar{q}) \geq f'(q_2) $ and $f'(q) \geq f'(\bar{q}) =0$ for $q \geq \bar{q}.$ Since $f'$ is decreasing in $[0,q_2],$ there is no $q_1 \in (0, q_2)$ such that $f'(q_1) =0.$ Thus $f$ must increase in $[0,q_1]$ for $q_1 < q_2 < \bar{q}$, decrease in $[q_1, \bar{q}]$, and increase on $[\bar{q}, 1/(1-\alpha)].$ In particular, $f(q) \geq f(\bar{q}) $ for all $q \in [0,1].$ Since $f(\bar{q}) = 0$, we get the desired inequality.
\end{proof}
\begin{lemma}
\label{lemma:compare_KL}
Let $p, q$ be distributions on $[n]$ satisfying $\alpha p + (1-\alpha)q = \vec{1}/n$ for $\alpha = k/n$ and $\alpha < 1/100$. Then for any $\mu \in \R_{[0,1]}^n$ with $\norm{\mu}_1 = k$ and $\mu^{\max} := \max_{i\in[n]} \mu_i < 1/100$,
\begin{align}
\DKL*{q \river \frac{\vec{1} - \mu}{n-k}} \le C_{\alpha,\mu^{\max}} \DKL*{p \river \frac{ \mu}{k}}
\end{align}
with $C_{\alpha, \mu^{\max}} := \frac{\alpha}{(1-\alpha)\log(1/(e\mu^{\max}))}$ as in \cref{lemma:tangentline}.
\end{lemma}
\begin{proof}
By \cref{lemma:tangentline} (for the choice $p = np_i$ and $q = nq_i$) and the fact that $C_{\alpha,\mu}$ is monotonically increasing in $\mu$, we deduce that
\Tag{
	\[ C_{\alpha,\mu^{\max}} p_i \log(p_ik/\mu_i) - q_i \log(q_i(n-k)/(1-\mu_i)) \ge \frac{1}{n}K_{\alpha,\mu^{\max}}(\alpha np_i - \mu_i). \]
}
\Tag<ieeetran>{
	\begin{multline*} C_{\alpha,\mu^{\max}} p_i \log(p_ik/\mu_i) - q_i \log(q_i(n-k)/(1-\mu_i)) \geq\\ \frac{1}{n}K_{\alpha,\mu^{\max}}(\alpha np_i - \mu_i). \end{multline*}
}
Summing this over all $i$ gives us
\Tag{
	\[ C_{\alpha,\mu^{\max}} \sum_{i\in[n]} p_i \log(p_ik/\mu_i) - \sum_{i\in[n]} q_i \log(q_i(n-k)/(1-\mu_i)) \ge \frac{1}{n}K_{\alpha,\mu^{\max}}\sum_{i\in[n]} (\alpha np_i - \mu_i). \]
}
\Tag<ieeetran>{
	\begin{multline*} C_{\alpha,\mu^{\max}} \sum_{i\in[n]} p_i \log(p_ik/\mu_i) -\\ \sum_{i\in[n]} q_i \log(q_i(n-k)/(1-\mu_i)) \ge \frac{1}{n}K_{\alpha,\mu^{\max}}\sum_{i\in[n]} (\alpha np_i - \mu_i). \end{multline*}
}
The r.h.s.\ equals $0$, so we deduce the desired inequality.
\end{proof}
Now, \cref{thm:one level contraction} is an easy corollary of \cref{lemma:compare_KL}.

\begin{proof}[Proof of \cref{thm:one level contraction}]
Let $\nu, \mu$ be the complement of $\bnu, \bmu$ resp.
Let $p = \nu D_{k \to 1} , q = \bnu D_{(n-k) \to 1}, $ and $\widehat{\mu}_i := \P_{F \sim \mu} {i \in F} $ for $i\in [n]$ in the setting of \cref{lemma:compare_KL}. Note that $\alpha p + (1-\alpha) q = \vec{1}/n$ for $\alpha = k/n,$ and $\widehat{\mu}^{\max} = p(\mu)^{\max} \leq \frac{1}{100}.$

Because $\mu$ is $1$-entropically independent (\cref{lemma:1entropic}),
\Tag{
	\[ \DKL*{p \river \frac{\widehat{\mu}}{k}} = \DKL*{\nu D_{k\to 1} \river \mu D_{k\to 1}} \le \frac{1}{k} \DKL{\nu \river \mu} = \frac{1}{k} \DKL{\bnu \river \bmu}. \]
}
\Tag<ieeetran>{
	\begin{multline*} \DKL*{p \river \frac{\widehat{\mu}}{k}} = \DKL*{\nu D_{k\to 1} \river \mu D_{k\to 1}} \le\\ \frac{1}{k} \DKL{\nu \river \mu} = \frac{1}{k} \DKL{\bnu \river \bmu}. \end{multline*}
}
Combining this with \cref{lemma:compare_KL} for $\alpha = k/n$ gives us
\begin{align*} \DKL*{q \river \frac{\vec{1}-\widehat{\mu}}{n-k}} &\le C_{\alpha,\widehat{\mu}^{\max}} \DKL*{p \river \frac{\widehat{\mu}}{k}}\\ &\le \frac{\frac{\alpha}{k} \DKL{\bnu \river \bmu}}{(1-\alpha)\log((ep(\mu)^{\max})^{-1})} \\ &= \frac{\DKL{\bnu \river \bmu}}{(n-k) \log ((e p(\mu)^{\max})^{-1})  }.\qedhere \end{align*}
\end{proof}

We can almost directly combine \cref{thm:one level contraction} and the average-case local-to-global principle \cref{thm:entropy contraction} to deduce an entropy contraction for $D_{(n-k)\to(n-k'+1)}$ and $D_{(n-k)\to(n-k-1)}$. The one remaining issue is that the local-to-global theorem requires that the marginals of \emph{conditionals of $\mu$} also have almost uniform marginals. This is the main result of \cref{sec:concentration}, which we state here.
\begin{theorem}
\label{thm:concentration}
Let $\mu \in \R^{\binom{[n]}{k}}$ be a strongly Rayleigh distribution, and let $T \subseteq [n]$ with $\card{T} = \bar{k}.$ For a sufficiently large constant $C$ and any $s \ge C(np(\mu)^{\max}+\bar{k})\log n$, we have
\[ \P*_{\substack{S \sim [n]\backslash T \\ \card{S} = n-s}}{ p(\mu_{[n]\backslash S})^{\max} \ge \frac{2p(\mu)^{\max} n}{s}} \le n^{-10}. \]
\end{theorem}
We prove this in \cref{sec:concentration}. We now have all the pieces to show our main technical result.

\begin{theorem}\label{thm:main technical}
Let $\mu: \binom{[n]}{k} \to \R_{\geq 0}$ be a strongly Rayleigh distribution with $p(\mu)^{\max}  \leq 1/500$. Let $s: = C(n p(\mu)^{\max} + \bar{k}) \log n$ for $C$ be as in \cref{thm:concentration} and $k' = \Theta(np(\mu)^{\max}).$
Then for any distribution $\bnu \subseteq \R^{\binom{[n]}{n-k}}$ and $\bar{k} \geq k+2$
\Tag{
	\[\DKL{\bnu D_{(n-k)\to (n-\bar{k}+1 )} \river \bmu D_{(n-k)\to (n-\bar{k}+1)}} \leq (1 - \kappa) \DKL{\bnu \river \bmu}\]
}
\Tag<ieeetran>{
	\begin{multline*}\DKL{\bnu D_{(n-k)\to (n-\bar{k}+1 )} \river \bmu D_{(n-k)\to (n-\bar{k}+1)}} \leq\\ (1 - \kappa) \DKL{\bnu \river \bmu}\end{multline*}
}
with $\kappa = \frac{\bar{k}-k-1}{2s \log n }.$ In particular,
\Tag{
	\[\DKL{\bnu D_{(n-k)\to (n-k'+1 )} \river  \bmu D_{(n-k)\to (n-k'+1)}} \leq (1 - \kappa_1) \DKL{\bnu \river \bmu} \] 
}
\Tag<ieeetran>{
	\begin{multline*}\DKL{\bnu D_{(n-k)\to (n-k'+1 )} \river  \bmu D_{(n-k)\to (n-k'+1)}} \leq\\ (1 - \kappa_1) \DKL{\bnu \river \bmu} \end{multline*}
}
and
\Tag{
	\[\DKL {\bnu D_{(n-k)\to (n-k-1)} \river  \bmu D_{(n-k)\to (n-k-1)}} \leq (1 - \kappa_2) \DKL{\bnu \river \bmu} \]
}
\Tag<ieeetran>{
	\begin{multline*}\DKL {\bnu D_{(n-k)\to (n-k-1)} \river  \bmu D_{(n-k)\to (n-k-1)}} \leq\\ (1 - \kappa_2) \DKL{\bnu \river \bmu} \end{multline*}
}
with $\kappa_1^{-1} = O(\log^2 n)$ and $\kappa_2^{-1} = O(np(\mu)^{\max} \log^2 n).$
\end{theorem}
\begin{proof}
Fix $\bar{k} \geq k+2$ to be chosen later, and a set $\bar{ T} \subseteq [n]$ of size $ n-\bar{k}.$ Let $s: = C(n p(\mu)^{\max} + \bar{k}) \log n$ for $C$ be as in \cref{thm:concentration}. In the context of \cref{thm:entropy contraction}, we want to bound $\gamma_{\bar{T}}$ with respect to $\bmu$. \cref{thm:one level contraction} implies that the link of $\emptyset$ is $1/u_0$-entropically independent with $u_0 = \log ((e p(\mu)^{\max})^{-1})$. Consider a random permutation $e_1, \dots, e_{n -\bar{k}}$ of elements of $\bar{T}.$ Note that each set $S_i:= \set*{e_1, \dots, e_i}$ is a randomly sampled size-$i$ subsets of $ \bar{T}.$ By using \cref{thm:concentration} and taking a union bound over $ i\in [n-s]$, we have that except with probability $ n^{-10} \times n = n^{-9},$ we have
\[p(\mu_{[n] \setminus S_i})^{\max} \leq \frac{2 p(\mu)^{\max} n}{ \abs{[n] \setminus S_i} } \leq \frac{2p(\mu)^{\max}}{s} \le \frac{2}{C \log n } < \frac{1}{100}\]
for $C$ sufficiently large. Suppose this event holds.
Note that the complement of $\mu_{[n] \setminus S_i}$ is exactly $ \bmu^{S_i}.$ Thus, \cref{thm:one level contraction} implies  that this link is $1/u_i$-entropically independent where $u_i:= \log(\frac{n -i}{2e p(\mu)^{\max} n}). $

As a result
\begin{equation} \label{ineq:chain contraction}
    \begin{split}
        \Tag<ieeetran>{&}\rho(\emptyset) \rho(S_1) \cdots \rho(S_{n-s}) \Tag{&}\geq \prod_{i=0}^{n-s} \parens*{1 - \frac{1}{(n- k - i) u_i}} 
\\
&\overset{(i)}{\geq} \exp \parens*{- \sum_{i=0}^{n-s} \parens*{\frac{1}{(n- k - i) u_i}+ \frac{1}{(n- k - i)^2 u_i^2}} }\\
    &\overset{(ii)}{\geq}  \exp \parens*{- \sum_{i=0}^{n-s} \parens*{\frac{1}{(n- i ) u_i} + \frac{k+1}{(n- k - i)^2} } }
    \end{split}
\end{equation}
where in $(i)$ we use $1-x \geq \exp(-x -x^2)$ for $x \le 1/2$, in $(ii)$ we use
\Tag{
	\[\frac{1}{(n- k - i) u_i}  = \frac{1}{(n-i) u_i} + \frac{k}{(n-i) (n-k-i) u_i} \leq \frac{1}{(n-i) u_i} + \frac{k}{(n-k-i)^2}\]
}
\Tag<ieeetran>{
	\begin{multline*}\frac{1}{(n- k - i) u_i}  = \frac{1}{(n-i) u_i} + \frac{k}{(n-i) (n-k-i) u_i} \leq\\ \frac{1}{(n-i) u_i} + \frac{k}{(n-k-i)^2}\end{multline*}
}
Next, let $h: = 2e p(\mu)^{\max} n.$
\Tag{
	\[\sum_{i=0}^{n-s} \frac{1}{(n- i ) u_i}\leq \int_{s}^n \frac{1}{t \log \frac{t}{h} } dt = \int_{s/h}^{n/h} \frac{1}{t \log t} dt = \log\log(n/h) - \log\log(s/h).\]
}
\Tag<ieeetran>{
	\begin{multline*}\sum_{i=0}^{n-s} \frac{1}{(n- i ) u_i}\leq \int_{s}^n \frac{1}{t \log \frac{t}{h} } dt = \int_{s/h}^{n/h} \frac{1}{t \log t} dt =\\ \log\log(n/h) - \log\log(s/h).\end{multline*}
}
where in the final equality we use $ (\log\log t)' = \frac{1}{t \log t}.$ Similarly
\[
    \sum_{i=0}^{n-s} \frac{k+1}{(n- k - i)^2} \leq (k+1)\int_{s-k}^{n-k} \frac{1}{t^2 } dt =  \frac{(k+1) (n-s)}{(s-k) (n-k)  }
\]
which is $\leq O(1)$, where we use the fact that $s\geq 2 C k \log n \gg k.$ Thus, we have that the l.h.s.\ in \cref{ineq:chain contraction} $\ge \Omega(\frac{1}{\log(n/h)}).$

Moreover, for $i\in \set*{n-s+1, \cdots, n-\bar{k}}$, by 1-entropic independence of $\bmu$ and its links (\cref{lemma:1entropic}) we get the trivial bound $\rho (S_i) \geq 1-\frac{1}{n-k-i } .$ Thus
\Tag{
	\[\prod_{i = n-s+1}^{n - \bar{k}} \rho(S_i) \geq \prod_{i = n-s+1}^{n - \bar{k}} \parens*{1-\frac{1}{n-k-i }} =  \frac{\bar{k}-k -1 }{s-k -1 } \geq \frac{\bar{k}-k -1 }{s}\]
}
\Tag<ieeetran>{
	\begin{multline*}\prod_{i = n-s+1}^{n - \bar{k}} \rho(S_i) \geq \prod_{i = n-s+1}^{n - \bar{k}} \parens*{1-\frac{1}{n-k-i }} =\\  \frac{\bar{k}-k -1 }{s-k -1 } \geq \frac{\bar{k}-k -1 }{s}\end{multline*}
}
Thus, with probability at least $1-n^{-9},$ $\prod_{i=0}^{n-\bar{k}}\rho(S_i) \geq \frac{\bar{k}-k -1 }{s \log n}$. Otherwise we have a trivial lower bound of $1/n$ on the product of $\rho$s due to $1$-entropic independence of $\bmu$ and its links. Thus
\Tag{
	\[\gamma_{\bar{T}} = \E*{\prod_{i=0}^{n-\bar{k}}\rho(S_i)^{-1} }^{-1}\geq\parens*{(1-n^{-9})\cdot \frac{s \log n}{\bar{k}-k -1 }+n^{-9}\cdot n}^{-1} \geq  \frac{\bar{k}-k -1 }{2s\log n}. \]
}
\Tag<ieeetran>{
	\begin{multline*}\gamma_{\bar{T}} = \E*{\prod_{i=0}^{n-\bar{k}}\rho(S_i)^{-1} }^{-1}\geq \\ \parens*{(1-n^{-9})\cdot \frac{s \log n}{\bar{k}-k -1 }+n^{-9}\cdot n}^{-1} \geq  \frac{\bar{k}-k -1 }{2s\log n}. \end{multline*}
}
We are done with the general entropy contraction statement. Next, we prove entropy contraction for specific values of $\bar{k}.$
Plugging in $\bar{k} = k'$ and noting that for our choice of $k'$ and $s$,
\[\bar{k} - k -1 \geq \frac{1}{2}\bar{k} \geq \Omega\parens*{\frac{s}{\log n}} \]
implies the first result. Similarly, setting $\bar{k} = k+2$ gives the second result for our choice of $s$.
\end{proof}
    \section{Concentration of marginals}
\label{sec:concentration}

The goal of this section is to show concentration of marginal upper bounds (\cref{thm:concentration}) for random conditionals of strongly Rayleigh distributions. In \cref{sec:entropy}, this is applied in the context of an average-case local-to-global principle (\cref{thm:entropy contraction}) to deduce our main entropy contraction result (\cref{thm:main technical}). The proof uses the following simple observation about covariances in a set-valued distribution $\mu$ whose support contains only sets of identical size.
\begin{lemma}[Covariances of homogeneous distributions]
\label{lemma:rowsum}
For any distribution $\mu$ defined over identically-sized subsets of a ground set of elements $[n]$ and any element $i \in [n]$ we have
\[ \sum_{j \in [n]} \parens*{p(\mu)_ip(\mu)_j - \P_{F \sim \mu}{i,j\in F}} = 0. \]
\end{lemma}

We show \cref{thm:concentration} by analyzing the marginal of each coordinate $i \in [n]$ conditioned on it remaining in $[n] \backslash S$ via a stochastic process. Formally, fix $T$ as in \cref{thm:concentration} and a coordinate $i \in [n]$ (possibly in $T$).
\begin{definition}[Stochastic process]
\label{def:process}
For fixed $T \subseteq [n]$ with $\card{T} = \bar{k}$, $i \in [n]$, and $s \le n-\card{T}-1$,
let $\sigma$ be a random permutation of $[n] \backslash (T \cup \set{i})$. For $0 \le t \le n-s$ define $S_t = \set{\sigma(1), \sigma(2), \dots, \sigma(t)}.$ Define $S = S_{n-s}$ and $\mu^{(t)} := \mu_{[n] \backslash S_t}$.
\end{definition}
Note that $S_t$ is generated from $S_{t-1}$ by removing a random element in $[n]\backslash(T\cup\set{i}\cup S_{t-1}).$ Now we can analyze $p(\mu_{[n]\backslash S})_i = p(\mu^{(n-s)})_i$ by analyzing the stochastic process $p(\mu^{(t)})_i.$ We start by analyzing its drift.

\begin{lemma}[Expected drift]
\label{lemma:drift}
With the setup in \cref{def:process} and $0 \le t < n-s$, we have
\Tag{
	\[ \E*_{S_{t+1}}{p(\mu^{(t+1)})_i \given S_t} - p(\mu^{(t)})_i \le \frac{(1-p(\mu^{(t)})^{\max})^{-1}}{n-\bar{k}-1-t}p(\mu^{(t)})_i. \]
}
\Tag<ieeetran>{
	\begin{multline*} \E*_{S_{t+1}}{p(\mu^{(t+1)})_i \given S_t} - p(\mu^{(t)})_i \le\\ \frac{(1-p(\mu^{(t)})^{\max})^{-1}}{n-\bar{k}-1-t}p(\mu^{(t)})_i. \end{multline*}
}
\end{lemma}
\begin{proof}
By definition, we know that if $S_{t+1} = S_t \cup \set{j}$ for some $j \in [n]\backslash(T\cup\set{i}\cup S_t),$ then
\begin{align}
    p(\mu^{(t+1)})_i - p(\mu^{(t)})_i &= \frac{\P*_{F \sim \mu^{(t)}}{i \in F, j\notin F} }{1 - p(\mu^{(t)})_j} - p(\mu^{(t)})_j \nonumber \\
    &= \frac{p(\mu^{(t)})_ip(\mu^{(t)})_j - \P*_{F \sim \mu^{(t)}}{i,j \in F}}{1 - p(\mu^{(t)})_j}. \label{eq:conditionj}
\end{align}
Hence by \cref{eq:conditionj},
\begin{align}
 &\E*_{S_{t+1}}{p(\mu^{(t+1)})_i \given S_t} - p(\mu^{(t)})_i \nonumber \\
 =~& \frac{1}{n - \card*{T\cup\set{i}\cup S_t}}\Tag<ieeetran>{\times\nonumber \\ &}\sum_{j \in [n]\backslash(T\cup\set{i}\cup S_t)} \frac{p(\mu^{(t)})_ip(\mu^{(t)})_j - \P*_{F \sim \mu^{(t)}}{i,j \in F}}{1 - p(\mu^{(t)})_j}. \label{eq:kanye}
\end{align}
Because $\mu$ and hence $\mu^{(t)}$ is strongly Rayleigh, each numerator of the fractions in \cref{eq:kanye} is nonnegative for $j \neq i$, hence the expression in \cref{eq:kanye} is at most
\begin{align*}
&\frac{(1-p(\mu^{(t)})^{\max})^{-1}}{n-\bar{k}-1-t}\Tag<ieeetran>{\times\\ &} \sum_{j \in [n]\backslash\set{i}} \parens*{p(\mu^{(t)})_ip(\mu^{(t)})_j - \P*_{F \sim \mu^{(t)}}{i,j \in F}} \\
\Tag{=~& \frac{(1-p(\mu^{(t)})^{\max})^{-1}}{n-\bar{k}-1-t} \Tag<ieeetran>{\times\\ &}\parens*{p(\mu^{(t)})_i - p(\mu^{(t)})_i^2 + \sum_{j \in [n]} \parens*{ p(\mu^{(t)})_ip(\mu^{(t)})_j - \P*_{F \sim \mu^{(t)}}{i,j \in F}}} \\}
\overset{(i)}{\le}~& \frac{(1-p(\mu^{(t)})^{\max})^{-1}}{n-\bar{k}-1-t}p(\mu^{(t)})_i,
\end{align*}
where $(i)$ follows from \cref{lemma:rowsum}. This completes the proof.
\end{proof}
Now we analyze the variance/maximum change in $p(\mu^{(t)})_i$.
\begin{lemma}[Variance and maximum change]
\label{lemma:maxchange}
With the setup in \cref{def:process} and $0 \le t < n-s$, we have with probability $1$ conditioned on $S_t$ that
\begin{align} p(\mu^{(t+1)})_i - p(\mu^{(t)})_i \le \frac{p(\mu^{(t)})^{\max}}{1 - p(\mu^{(t)})^{\max}}p(\mu^{(t)})_i. \label{eq:maxchange} \end{align}
Also, we have the variance bound
\Tag{
	\[ \E*_{S_{t+1}}{\parens*{p(\mu^{(t+1)})_i - p(\mu^{(t)})_i}^2 \given S_t} \le \frac{1}{n-\bar{k}-1-t} \cdot \frac{p(\mu^{(t)})^{\max}}{(1 - p(\mu^{(t)})^{\max})^2}p(\mu^{(t)})_i^2. \label{eq:var} \]
}
\Tag<ieeetran>{
	\begin{multline*} \E*_{S_{t+1}}{\parens*{p(\mu^{(t+1)})_i - p(\mu^{(t)})_i}^2 \given S_t} \le\\ \frac{1}{n-\bar{k}-1-t} \cdot \frac{p(\mu^{(t)})^{\max}}{(1 - p(\mu^{(t)})^{\max})^2}p(\mu^{(t)})_i^2. \label{eq:var} \end{multline*}
}
\end{lemma}
\begin{proof}
By the formula in \cref{eq:conditionj} we get
\Tag{
	\[p(\mu^{(t+1)})_i - p(\mu^{(t)})_i = \frac{p(\mu^{(t)})_ip(\mu^{(t)})_j - \P*_{F \sim \mu^{(t)}}{i,j \in F}}{1 - p(\mu^{(t)})_j} \le \frac{p(\mu^{(t)})^{\max}}{1 - p(\mu^{(t)})^{\max}}p(\mu^{(t)})_i.\]
}
\Tag<ieeetran>{
	\begin{multline*}p(\mu^{(t+1)})_i - p(\mu^{(t)})_i =\\ \frac{p(\mu^{(t)})_ip(\mu^{(t)})_j - \P*_{F \sim \mu^{(t)}}{i,j \in F}}{1 - p(\mu^{(t)})_j} \le\\ \frac{p(\mu^{(t)})^{\max}}{1 - p(\mu^{(t)})^{\max}}p(\mu^{(t)})_i.\end{multline*}
}
Because $\mu$ and hence $\mu^{(t)}$ is strongly Rayleigh, $p(\mu^{(t+1)})_i \ge p(\mu^{(t)})_i$ so
\begin{align*}
&\E*_{S_{t+1}}{\parens*{p(\mu^{(t+1)})_i - p(\mu^{(t)})_i}^2 \given S_t} \\ \overset{(i)}{\le}~&\frac{p(\mu^{(t)})^{\max}}{1 - p(\mu^{(t)})^{\max}}p(\mu^{(t)})_i \E*_{S_{t+1}}{p(\mu^{(t+1)})_i - p(\mu^{(t)})_i \given S_t}
\\ \overset{(ii)}{\le}~&\frac{1}{n-\bar{k}-1-t} \cdot \frac{p(\mu^{(t)})^{\max}}{(1 - p(\mu^{(t)})^{\max})^2}p(\mu^{(t)})_i^2,
\end{align*}
where $(i)$ follows from \cref{eq:maxchange} and $(ii)$ follows from \cref{lemma:drift}.
\end{proof}
Our desired concentration bound now essentially follows from a careful application of Bernstein's inequality for martingales to the sequence $\log p(\mu^{(t)})_i.$
\begin{theorem}[{\cite[Theorem 20]{CL06}}]
\label{thm:bernstein}
Let $X^{(0)}, X^{(1)}, \dots, X^{(t)}$ be a martingale such that $X^{(u)} - X^{(u-1)} \le M$ with probability $1$ and $\Var*{X^{(u)} \given X^{(u-1)}} \le \sigma_u^2$ for $u \in [t]$. Then \[ \P*{X^{(t)} - X^{(0)} \ge \lambda} \le \exp\parens*{-\frac{\lambda^2}{2\sum_{u\in[t]}\sigma_u^2 + 2M\lambda/3} }. \]
\end{theorem}

\begin{proof}[Proof of \cref{thm:concentration}]
For the setup in \cref{def:process}, define the random variables $Y^{(t)} := \log p(\mu^{(t)})_i$ indexed by $0 \le t \le n-s$. Given this, we define the martingale $X^{(0)} = Y^{(0)}$ and \[ X^{(t+1)} := X^{(t)} + Y^{(t+1)} - Y^{(t)} - \E*_{S_{t+1}}{Y^{(t+1)} - Y^{(t)} \given S_t}, \]
for $0 \le t < n-s$.

We will bound the drift, maximum change, and variance of $Y^{(t+1)}$ assuming that $p(\mu^{(t)})^{\max} \le 2p(\mu)^{\max} n/(n-t) < 1/10$ (which we want to show holds with high probability). We may assume this, because we can just prematurely stop the stochastic process whenever this condition breaks. For the drift term, we bound
\begin{align}
&\E*_{S_{t+1}}{Y^{(t+1)} - Y^{(t)} \given S_t} \Tag<ieeetran>{\nonumber \\} = \Tag<ieeetran>{~&} \E*_{S_{t+1}}{\log(p(\mu^{(t+1)})_i/p(\mu^{(t)})_i) \given S_t} \nonumber \\
\le~& \E*_{S_{t+1}}{\frac{p(\mu^{(t+1)})_i - p(\mu^{(t)})_i}{p(\mu^{(t)})_i} \given S_t}\Tag<ieeetran>{\nonumber \\}\overset{(i)}{\le}\Tag<ieeetran>{~&} \frac{(1-p(\mu^{(t)})^{\max})^{-1}}{n-\bar{k}-1-t} \nonumber \\
\overset{(ii)}{\le}~& \frac{1}{n-\bar{k}-1-t} + \frac{4p(\mu)^{\max} n}{(n-\bar{k}-1-t)(n-t)},
\label{eq:drifty}
\end{align}
where $(i)$ follows from \cref{lemma:drift} and $(ii)$ follows from our assumption on $p(\mu^{(t)})^{\max}$. For the bound on the maximum increase, we get
\begin{align}
    Y^{(t+1)} - Y^{(t)} &= \log\parens*{p(\mu^{(t+1)})_i/p(\mu^{(t)})_i} \Tag<ieeetran>{\nonumber \\&}\le \frac{p(\mu^{(t+1)})_i - p(\mu^{(t)})_i}{p(\mu^{(t)})_i} \nonumber \\
    &\le p(\mu^{(t+1)})_i - p(\mu^{(t)})_i \le \frac{p(\mu^{(t)})^{\max}}{1 - p(\mu^{(t)})^{\max}} \Tag<ieeetran>{\nonumber \\&}\le \frac{4p(\mu)^{\max}n}{n-t}.
    \label{eq:maxy}
\end{align}
by \cref{lemma:maxchange} \cref{eq:maxchange} and our assumption on $p(\mu^{(t)})^{\max}$. For the variance term, we first bound
\begin{align}
    &\E*_{S_{t+1}}{\parens*{Y^{(t+1)} - Y^{(t)}}^2 \given S_t} \Tag<ieeetran>{\nonumber \\}=\Tag<ieeetran>{~&} \E*_{S_{t+1}}{\log\parens*{\mu^{(t+1)})_i/p(\mu^{(t)})_i}^2 \given S_t} \nonumber \\
    \le~& \E*_{S_{t+1}}{\parens*{\frac{p(\mu^{(t+1)})_i - p(\mu^{(t)})_i}{p(\mu^{(t)})_i}}^2 \given S_t} \nonumber \\
    \le~& \frac{1}{n-\bar{k}-1-t} \cdot \frac{p(\mu^{(t)})^{\max}}{(1 - p(\mu^{(t)})^{\max})^2} \Tag<ieeetran>{\nonumber \\}\le\Tag<ieeetran>{~&} \frac{8np(\mu)^{\max}}{(n-\bar{k}-1-t)(n-t)} \label{eq:vary}
\end{align}
by \cref{lemma:maxchange} and our assumption on $p(\mu^{(t)})^{\max}$. Our next goal is to prove that $X^{(t)} \le X^{(0)} + 1/10$ with high probability by using \cref{thm:bernstein}. Because $\mu$ and hence $\mu^{(u)}$ is strongly Rayleigh for all $0 \le u \le t$,
\Tag{
	\[ X^{(u+1)} - X^{(u)} \le Y^{(u+1)} - Y^{(u)} \le \frac{4p(\mu)^{\max}n}{n-u} \le \frac{4p(\mu)^{\max}n}{s} \le \frac{1}{1000\log n} \]
}
\Tag<ieeetran>{
	\begin{multline*} X^{(u+1)} - X^{(u)} \le Y^{(u+1)} - Y^{(u)} \le \frac{4p(\mu)^{\max}n}{n-u} \le\\ \frac{4p(\mu)^{\max}n}{s} \le \frac{1}{1000\log n} \end{multline*}
}
by \cref{eq:maxy} and sufficiently large $C$ for $s \ge C(np(\mu)^{\max}+\bar{k})\log n$ so we may take $M = 1/(1000\log n)$ in \cref{thm:bernstein}. Additionally, we have that
\begin{align*}
\Var*{X^{(u+1)} \given X^{(u)}} &\le \E*_{S_{t+1}}{\parens*{Y^{(t+1)} - Y^{(t)}}^2 \given S_t} \\
&\le \frac{8np(\mu)^{\max}}{(n-\bar{k}-1-u)(n-u)} \Tag<ieeetran>{\\ &}\le \frac{8np(\mu)^{\max}}{(n-\bar{k}-1-u)^2},
\end{align*}
by \cref{eq:vary} so we may take $\sigma_u^2 =8np(\mu)^{\max}/(n-\bar{k}-1-u)^2$ in \cref{thm:bernstein}. By \cref{thm:bernstein} for $\lambda = 1/10$, $M = 1/(1000\log n)$, and $\sigma_u^2 =8np(\mu)^{\max}/(n-\bar{k}-1-u)^2$, we get that
\begin{align*}
    &\P*{X^{(t)} - X^{(0)} \ge 1/10} \le \exp\parens*{-\frac{1/100}{2\sum_{u\in[t]}\sigma_u^2 + M/15}} \\
    \le&\exp\parens*{-\frac{1/100}{16np(\mu)^{\max}\sum_{u \le n-s} \frac{1}{(n-\bar{k}-1-u)^2} + \frac{1}{15000\log n}}} \Tag<ieeetran>{\\}\le\Tag<ieeetran>{&} \exp\parens*{-\frac{1/100}{\frac{50np(\mu)^{\max}}{s-\bar{k}-1} + \frac{1}{15000\log n}}} \\
    \le& \exp(-20\log n) = n^{-20}
\end{align*}
for sufficiently large $C$ in $s \ge C(np(\mu)^{\max}+\bar{k})\log n \ge C\bar{k}\log n$. To finish, note that
\begin{align} \Tag<ieeetran>{&} Y^{(t)} - Y^{(0)} \Tag<ieeetran>{\nonumber \\} &= X^{(t)} - X^{(0)} + \sum_{u \in [t-1]} \E*_{S_{u+1}}{Y^{(u+1)} - Y^{(u)} \given S_u} \nonumber \\
&\overset{(i)}{\le} X^{(t)} - X^{(0)} + \Tag<ieeetran>{\nonumber \\ &~}\sum_{u \in [t-1]} \frac{1}{n-\bar{k}-1-u} + \frac{4p(\mu)^{\max} n}{(n-\bar{k}-1-u)(n-u)} \nonumber \\
&\overset{(ii)}{\le} X^{(t)} - X^{(0)} + \log\parens*{\frac{n-\bar{k}-1}{n-\bar{k}-1-t} } + \frac{20p(\mu)^{\max} n}{n-\bar{k}-1-t} \nonumber \\
&\overset{(iii)}{\le} X^{(t)} - X^{(0)} + \log\parens*{\frac{n}{n-t}} + \frac{1}{100\log n}, \label{eq:ytox}
\end{align}
where $(i)$ follows from \cref{eq:drifty}, $(ii)$ follows from direction calculations with Riemann integrals, and $(iii)$ follows from $t \le n-s$ and $s \ge C(np(\mu)^{\max}+\bar{k})\log n$. To conclude, we write
\begin{align*}
    &\P*{p(\mu^{(t)})_i \ge \frac{2p(\mu)_in}{n-t}} \Tag<ieeetran>{\\}=\Tag<ieeetran>{&} \P*{Y^{(t)} - Y^{(0)} \ge \log\parens*{\frac{n}{n-t}} + \log 2} \\
    \overset{(i)}{\le}& \P*{X^{(t)} - X^{(0)} \ge 1/10} \le n^{-20},
\end{align*}
where $(i)$ follows from \cref{eq:ytox}. \cref{thm:concentration} now follows from union-bounding over all times $t \in [n-s]$ and all coordinates $i \in [n].$
\end{proof}
    \section{Isotropic rounding}
\label{sec:estimation}

We give a reduction which estimates marginals of a distribution given an algorithm that samples using marginal overestimates. At a high level, we split our original strongly Rayleigh distribution $\mu \in \R^{\binom{[n]}{k}}$ into two smaller distributions (supported on $S_1, S_2$ for $[n] = S_1 \sqcup S_2$), and recursively produce marginal overestimates for $\mu_{S_1}$ and $\mu_{S_2}$ that sum to at most $4k$. Now, we merge these groups. Because $\mu$ is strongly Rayleigh, the marginal overestimates on $\mu_{S_1}$ and $\mu_{S_2}$ provide marginal overestimates for $\mu$ summing to at most $8k$. Thus, we can cheaply take $O(n\log n/k)$ samples from $\mu$ to get marginal overestimates summing to at most $4k$ again.

\begin{theorem}[Isotropic rounding from sampling]
\label{lemma:marginalestimate}
Let $\mu \in \R^{\binom{[n]}{k}}$ be a strongly Rayleigh distribution. Assume that we can sample from restrictions $\mu_S$ of $\mu$ (\cref{def:conditional}) in time $\mathcal{A}_{\mu}(K)$ given marginal overestimates of $\mu_S$ that sum to at most $K$\footnote{We do not write an $n$ dependence as it will be polylogarithmic in our algorithms.}. Then there is an algorithm that produces marginal overestimates $q_i \ge p(\mu)_i$ with sum $\sum_{i\in[n]} q_i \le 4k$ in time $\Otilde*{n/k \cdot \mathcal{A}_{\mu}(8k)}$.
\end{theorem}
\begin{proof}
We use a divide-and-conquer procedure. Precisely, given a set $S$ we use the following algorithm to produce marginals overestimates of $S$ summing to at most $4k$. If $\card{S} \le 4k$, then we let all our overestimates be $1$. Otherwise, partition $S = S_1 \sqcup S_2$ into equally sized pieces, and recursively produce marginal overestimates $q_i^{(1)} \ge p(\mu_{S_1})_i$ and $q_i^{(2)} \ge p(\mu_{S_2})_i$ with $\sum_{i\in S_1} q_i^{(1)} \le 4k$ and $\sum_{i\in S_2} q_i^{(2)} \le 4k$.

Because $\mu$ is strongly Rayleigh, in fact $q_i^{(1)} \ge p(\mu_S)_i$ for all $i \in S_1$ and $q_i^{(2)} \ge p(\mu_S)_i$ for all $i \in S_2.$ Hence the vector $\bar{q} \in \R^S$ defined as $\bar{q}_i = q_i^{(1)}$ for $i \in S_1$ and $\bar{q}_i = q_i^{(2)}$ for $i \in S_2$ are marginal overestimates for $\mu_S$. Additionally, $\sum_{i\in S} \bar{q}_i \le 8k.$

Set $s = \frac{100\card{S} \log n}{k},$ and let $F_1, F_2, \dots, F_s$ be independent samples from $\mu_S$ generated in total time $\Otilde{s \cdot \mathcal{A}_{\mu}(8k)},$ by using the overestimates $\bar{q}.$ Define for $i \in S$
\[ q_i = \max\set*{\frac{k}{\card{S}}, \frac{2\card*{\set{s' \in [s] \given i \in F_{s'}}}}{s}}. \]
We claim that $q_i$ are marginal overestimates for $\mu_S$ with high probability and sum to at most $4k$. The sum follows because
\Tag{
	\[ \sum_{i\in S} q_i \le \sum_{i\in S} \parens*{\frac{k}{\card{S}} + \frac{2\card*{\set{s' \in [s] \given i \in F_{s'}}}}{s} } \le k + \sum_{s' \in [s]} \frac{2\card{F_{s'}}}{s} = 3k. \]
}
\Tag<ieeetran>{
	\begin{multline*} \sum_{i\in S} q_i \le \sum_{i\in S} \parens*{\frac{k}{\card{S}} + \frac{2\card*{\set{s' \in [s] \given i \in F_{s'}}}}{s} } \le\\  k + \sum_{s' \in [s]} \frac{2\card{F_{s'}}}{s} = 3k. \end{multline*}
}
Now we show that $q_i$ are marginal overestimates of $\mu_S$. The case $p(\mu_S)_i \le k/\card{S}$ is trivial. Otherwise, by a Chernoff bound,
\Tag{
	\[ \P*_{F_1,\dots,F_s}{\card*{\set{s' \in [s] \given i \in F_{s'}}} \le p(\mu_S)_is/2 } \le \exp(-p(\mu_S)_is/8) \le n^{-100} \]
}
\Tag<ieeetran>{
	\begin{multline*} \P*_{F_1,\dots,F_s}{\card*{\set{s' \in [s] \given i \in F_{s'}}} \le p(\mu_S)_is/2 } \le\\ \exp(-p(\mu_S)_is/8) \le n^{-100} \end{multline*}
}
by the choice $s = \frac{100\card{S}\log n}{k}$ and $p(\mu_S)_i \ge k/\card{S}.$

The final runtime claim follows from recursively calling the above describde algorithm on $\mu$, and using that there are $O(\log n)$ layers, each with total size $n$. Precisely, the total number of samples taken in a layer is at most the sum over $s = \frac{100\card{S}\log n}{k}$ in a layer, which is $O(n/k \cdot \log n).$
\end{proof}
    \section{Proofs of main results}
\label{sec:mainproof}

In this section, we combine our previous results to show \cref{thm:overestimate,thm:updown,cor:spanning,cor:dpp}. All results will follow from our entropy contraction bound \cref{thm:main technical} combined with \cref{lem:entropy-contraction-implies-mlsi}. Also, \cref{thm:updown} requires \cref{lemma:marginalestimate}.

\begin{proof}[Proof of \cref{thm:overestimate}]
We first apply \cref{prop:near-isotropic} to instead focus on sampling from a strongly Rayleigh distribution $\mu'$ with all marginals bounded by $K/n$. Let $\bar{\mu'}$ be the complement of $\mu'$.

For $\kappa_1$ as in \cref{thm:main technical} we run $O(\kappa_1^{-1}\log n)$ steps of a down-up operator on the complement of our set, to converge to $\bar{\mu'}$, i.e., iterate the Markov chain $D_{(n-k)\to(n-k'+1)}U_{(n-k'+1)\to(n-k)}$. Note that each $U_{(n-k'+1)\to(n-k)}$ part needs to be implemented via a baseline sampling algorithm which takes time $\mathcal{T}(k'-1,k)$. By \cref{thm:main technical,lem:entropy-contraction-implies-mlsi}, and the fact that the up step $U_{(n-k'+1)\to(n-k)}$ cannot increase entropy, the chain mixes in $O(\kappa_1^{-1}\log n) = O(\log^3 n)$ steps. Thus the runtime is bounded by making $O(\log^3 n)$ calls to $\mathcal{T}_{\mu}(O(K), k)$ as $k' = \Theta(np(\mu')^{\max}) = O(K)$ from \cref{thm:main technical}.
\end{proof}

\begin{proof}[Proof of \cref{cor:spanning}]
Let $n = \card{E}, k = \card{V}$. By the results of \cite{ALOVV21} a spanning tree can be sampled in $\Otilde{\card{E}}$ time on a graph with edge set $E$. Hence applying \cref{thm:overestimate} with $K = O(k) = \Otilde{\card{V}}$ shows that after obtaining the initial overestimates, each future sample requires $\Otilde{\card{V}}$ time. Obtaining the original overestimates takes time $\Otilde{n/k \cdot k} = \Otilde{\card{E}}$ by \cref{lemma:marginalestimate}.
\end{proof}

\begin{proof}[Proof of \cref{cor:dpp}]
Let $\mu$ be the $k$-DPP with ensemble matrix $L.$
Sampling from a $k$-DPP over ground set of size $n$ can be done in $\Otilde{n^\omega}$ time, see \cref{lem:dpp-matmult}. This together with \cref{lemma:marginalestimate} shows that in $\tilde{O}(n/k\cdot k^\omega) = \Otilde{nk^{\omega-1}}$ time, we can get marginal overestimates that sum to $\tilde{O}(k),$ as well as one initial sample $S_0$ from $\mu.$  We now apply \cref{thm:overestimate}, and note that each oracle call is equivalent to sampling from a $k$-DPP on a size-$O(K)$ subset of $[n]$, which takes $\tilde{O}(k^\omega)$ time.
\end{proof}

We now state a result on sampling strongly Rayleigh distributions using up-down steps. While it is known that this Markov chain normally mixes in $\Otilde{n}$ steps \cite{CGM19}, we show that under isotropy of $\mu$ it mixes in $\Otilde{k}$ steps.
\begin{theorem}
\label{thm:updown}
Given a strongly Rayleigh distribution $\mu \in \R^{\binom{[n]}{k}}$ and marginal overestimates $q_i \ge \P_{T \sim \mu}{i \in T}$ for $i \in [n]$ which sum to $K := \sum_{i\in[n]} q_i,$ there is an algorithm that samples from a distribution with total variation distance $n^{-O(1)}$ from $\mu$ in time bounded by $\Otilde{K}$ calls to $\mathcal{T}_{\mu}(k+1, k)$. Additionally, given a strongly Rayleigh distribution $\mu \in \R^{\binom{[n]}{k}}$, we can produce marginal overestimates $q_i \ge \P_{T \sim \mu}{i \in T}$ with sum $\sum_{i\in [n]} q_i \le O(k)$ in time $\Otilde{n \cdot \mathcal{T}_{\mu}(k+1, k)}$.
\end{theorem}
\begin{proof}[Proof of \cref{thm:updown}]
For $\kappa_2$ as in \cref{thm:main technical} we run $O(\kappa_2^{-1}\log n)$ steps of the one level down-up operator on the complement of our set, to converge to $\bar{\mu'}$, i.e., iterate the Markov chain $D_{(n-k)\to(n-k-1)}U_{(n-k-1)\to(n-k)}.$ Note that each $U_{(n-k+1)\to(n-k)}$ part needs to be implemented via a baseline sampling algorithm which takes time $\mathcal{T}(k+1,k)$. By \cref{thm:main technical,lem:entropy-contraction-implies-mlsi}, and the fact that the up step $U_{(n-k-1)\to(n-k)}$ cannot increase entropy, the chain mixes in $O(\kappa_2^{-1}\log n) = O(K \log^3 n)$ steps. Thus the runtime is bounded by making $O(K \log^3 n)$ calls to $\mathcal{T}_{\mu}(k+1, k).$

Finally, by \cref{lemma:marginalestimate}, we can obtain the desired overestimates $q_i$ in time \[ \Otilde{n/k \cdot 8k \cdot \mathcal{T}_{\mu}(k+1, k)} = \Otilde{n \cdot \mathcal{T}_{\mu}(k+1, k)} \] as desired.
\end{proof}

    \Tag{\section{Deferred proofs} \label{sec:defer}

\begin{proof}[Proof of \cref{thm:entropy contraction}]
	Let $\nu$ be an arbitrary distribution. Let $f(x):=x\log x$ and note that 
	\[\DKL{\nu\river \mu}=\E_{S\sim \mu}{f(\nu(S)/\mu(S))}-f(\E_{S\sim \mu}{\nu(S)/\mu(S)}).\]

	Consider the following process: We sample a set $S\sim \mu$ and uniformly at random permute its elements to obtain $X_1,\dots,X_k$. Define the random variable
	\Tag{
		\[ \tau_i=f\parens*{\E*{\frac{\nu(S)}{\mu(S)}\given X_1,\dots,X_i}}=f\parens*{\frac{\sum_{S'\ni X_1,\dots, X_i}\nu(S')}{\sum_{S'\ni X_1,\dots,X_i}\mu(S')}}=f\parens*{\frac{\nu D_{k\to i}(\set{X_1,\dots,X_i})}{\mu D_{k\to i}(\set{X_1,\dots,X_i})}}, \]
	}
	\Tag<ieeetran>{
		\begin{multline*} \tau_i=f\parens*{\E*{\frac{\nu(S)}{\mu(S)}\given X_1,\dots,X_i}}=\\ f\parens*{\frac{\nu D_{k\to i}(\set{X_1,\dots,X_i})}{\mu D_{k\to i}(\set{X_1,\dots,X_i})}}, \end{multline*}
	}
	Note that $\tau_i$ is a ``function'' of $X_1,\dots,X_i$. It is not hard to see that 
	\[\DKL{\nu\river \mu}=\E{\tau_k}-\E{\tau_0}=\sum_{i=0}^{k-1}\E{\tau_{i+1}-\tau_i}. \]
	Conveniently, we obtain $\DKL{\nu D_{k\to \l}\river \mu D_{k\to\l}}$ by just summing over the first $\l$ terms:
	\[ \DKL{\nu D_{k\to \l}\river \mu D_{k\to \l}}=\E{\tau_\l}-\E{\tau_0}=\sum_{i=0}^{\l-1}\E{\tau_{i+1}-\tau_i}. \]
	Our goal is to show that the sum of the last $k-\l$ terms are at least $\kappa$ times the entire sum. Applying the assumption of local contraction to the link of the set $T=\set{X_1,\dots,X_i}$, we get
	\[ \E{\tau_{i+1}-\tau_i\given X_1,\dots,X_i}\leq (1-\rho(T))\cdot\E{\tau_k-\tau_i\given X_1,\dots,X_i}, \]
	which we rewrite as
	\[ \E{\tau_k-\tau_{i+1}\given X_1,\dots,X_i}\geq \rho(T)\cdot \E{\tau_k-\tau_i\given X_1,\dots, X_i}. \]
	Define the random variable
	\[ Z_i:=\frac{\tau_k-\tau_i}{\rho(\emptyset)\rho(\set{X_1})\cdots \rho(\set{X_1,\dots,X_{i-1}})}, \]
	and note that our previous inequality simplifies to $\E{Z_{i+1}\given X_1,\dots,X_i}\geq \E{Z_i\given X_1,\dots, X_i}$. Chaining these inequalities we get that $\E{Z_\l}\geq \E{Z_0}=\DKL{\nu\river \mu}$. We can further simplify $\E{Z_\l}$ by noting that the numerator $\tau_k-\tau_\l$ remains the same if we permute $X_1,\dots,X_\l$.
	\Tag{
		\[ \E{Z_\l\given \set{X_1,\dots,X_\l}, \set{X_{\l+1},\dots,X_k}}= \gamma(\set{X_1,\dots,X_\l})\cdot (\tau_k(\set{X_1,\dots,X_k})-\tau_\l(\set{X_1,\dots,X_\l})). \]
	}
	\Tag<ieeetran>{
		\begin{multline*} \E{Z_\l\given \set{X_1,\dots,X_\l}, \set{X_{\l+1},\dots,X_k}}=\\ \gamma(\set{X_1,\dots,X_\l})\cdot (\tau_k(\set{X_1,\dots,X_k})-\tau_\l(\set{X_1,\dots,X_\l})). \end{multline*}
	}
	Taking a further expectation, we get
	\[ \E{Z_\l}\leq \min\set*{\gamma(T)\given T\in \binom{[n]}{\l}}\cdot \E{\tau_k-\tau_\l}. \]
	This together with $\E{Z_\l}\geq \E{Z_0}=\E{\tau_k-\tau_0}$ completes the proof.
\end{proof}

\begin{proof}[Proof of \cref{prop:near-isotropic}]
Clearly, $\P_{S \sim \mu'}{i^{(j)} \in S} \le p_i/t_i \le K/n$.
Also, \[ \card{U} = \sum_{i\in[n]} t_i \le \sum_{i\in[n]} \parens*{1 + \frac{n}{K}p_i} \le n+n\cdot \frac{\sum_i p_i}{K}\leq 2n. \]

For the third property, if $\mu$ has the generating polynomial $g_{\mu}(z_1, \dots, z_n),$ then the distribution $\mu'$ obtained by subdividing element $i$ into $t_i$ copies has generating polynomial
\Tag{
	\[g_{\mu'}(z_1^{(1)}, \ldots, z_n^{(t_n)}) = g_\mu \parens*{\frac{z_1^{(1)} + \ldots + z_1^{(t_1)}}{t_1}, \ldots, \frac{z_n^{(1)} + \ldots + z_n^{(t_n)}}{t_n} }.\]
}
\Tag<ieeetran>{
	\begin{multline*}g_{\mu'}(z_1^{(1)}, \ldots, z_n^{(t_n)}) =\\ g_\mu \parens*{\frac{z_1^{(1)} + \ldots + z_1^{(t_1)}}{t_1}, \ldots, \frac{z_n^{(1)} + \ldots + z_n^{(t_n)}}{t_n} }.\end{multline*}
}
Clearly, if $g_{\mu}$ is real-stable then so is $g_{\mu'}$. This is because if $z_i^{j}$ are chosen from the upper half plane $\set{z\in \C\given \Im(z)>0}$, their averages also lie in the upper half plane.
\end{proof}

For the convenience of the reader, we prove \cref{lem:entropy-contraction-implies-mlsi}. The proof is standard and closely related to the standard proofs bounding the mixing time of a Markov chain based on the modified log-Sobolev constant \cite{BT06}. It is a well-known fact that entropy contraction by a factor $1-\alpha$ implies an MLSI inequality of the same magnitude.

\begin{lemma}\label{lem:dpp-matmult}
    Given an $n\times n$ positive semidefinite matrix $L$ and an integer $k\leq n$, we can sample from the $k$-DPP defined by $L$ in time $\Otilde{n^\omega}$.
\end{lemma}
We remark that variants of this statement where slow (but more practical) matrix multiplication algorithms are used, which result in cubic $\Otilde{n^3}$ runtimes, already exist in the literature. Here, we simply formalize the observation that these algorithms can be adapted to take advantage of fast matrix multiplication and thus the runtime can be reduced to $\Otilde{n^\omega}$.\footnote{See p.18 in \url{https://buildmedia.readthedocs.org/media/pdf/dppy/latest/dppy.pdf} for details on various algorithms for sampling from DPPs.}
\begin{proof}
    \Textcite{KT12} reduce the task of sampling from a $k$-DPP defined by $L$ to sampling from a (size-unconstrained) DPP. This is achieved by performing a spectral decomposition of the kernel matrix, choosing a subset of exactly $k$ eigenvectors, each subset chosen with probability proportional to the product of the corresponding eigenvalues and forming a new kernel matrix just from the chosen eigenvectors. For details, see \cite{KT12}. We simply remark that an approximate spectral decomposition of $L$ is the most expensive operation here (while choosing the subset of eigenvectors can be done in $O(n^2)$ time). Thus, this part of the algorithm takes time $\Otilde{n^\omega}$ using fast matrix multiplication \cite{LY93,BVKS20}.
    
    Now, for sampling from a (size-unconstrained) DPP, \textcite{KT12} presented a somewhat slow $O(n^4)$-time algorithm, which was subsequently refined to $O(n^3)$, see, e.g., \cite{Poul20}. The same algorithm can be improved by switching linear algebraic operations it uses to those that employ fast matrix multiplication. The factorization-based algorithm presented by \textcite{Poul20} arranges the ground set of $n$ elements as leaves of a balanced binary tree, where the final sample from the DPP is produced at the root of the tree. Each node of this binary tree with $m$ leaves in its subtree is associated with an $m\times m$ kernel matrix. Roughly speaking, a node with $m$ leaves first computes an $m/2\times m/2$ submatrix for its left child (the marginal of its DPP on the first half of the elements), produces a sample from the left subtree, and then produces another $m/2\times m/2$ submatrix for its right child (the conditional DPP, conditioned on choices made by the first child). These submatrices are produced simply by Schur complements and matrix multiplication, all of which take time $\Otilde{m^\omega}$ using fast matrix multiplication. Summing over all levels of the binary tree results in an overall runtime of $\Otilde{n^\omega}$. 
\end{proof}}
	\PrintBibliography
\end{document}